\title{Reputation Effects under Short Memories}
\author{Harry PEI\footnote{Email: harrydp@northwestern.edu. I thank Dilip Abreu, Dan Barron, James Best, Joyee Deb, Laura Doval, Jeff Ely, Drew Fudenberg, George Georgiadis, Yingkai Li, Qingmin Liu, Daniel Luo, Lucas Maestri, Meg Meyer,
Wojciech Olszewski, Alessandro Pavan, Larry Samuelson, Ali Shourideh, Vasiliki Skreta, Andrzej Skrzypacz,
Alex Smolin, Takuo Sugaya, Caroline Thomas, Juuso V\"{a}lim\"{a}ki, Allen Vong, and Alex Wolitzky for helpful comments. I thank NSF Grant SES-1947021 for financial support.}}
\date{\today}
\documentclass[11pt]{article}

\usepackage{amsmath}
\usepackage{graphicx}
\usepackage{amsfonts}
\usepackage{amsthm}
\usepackage{setspace}
\usepackage{tikz}
\usepackage{amssymb}
\usetikzlibrary{patterns}
\usepackage{sgame}
\usepackage{color}
\usepackage{hyperref}
\usepackage{times}
\usepackage{enumitem}
\usepackage{csquotes}
\usepackage{multirow,array}
\usepackage[top=1in, bottom=1in, left=1in, right=1in]{geometry}
\begin{document}
\numberwithin{equation}{section}

\maketitle

\noindent \textbf{Abstract:} I analyze a novel reputation game between a patient seller and a sequence of myopic consumers, in which the consumers have limited memories and do not know the exact sequence of the seller's actions. I focus on the case where each consumer only observes the number of times that the seller took each of his actions in the last $K$ periods. When payoffs are monotone-supermodular, I show that the patient seller can approximately secure his commitment payoff in all equilibria as long as $K$ is at least one. I also show that the consumers can approximately attain their first-best welfare in all equilibria \textit{if and only if} their memory length $K$ is lower than some cutoff. Although a longer memory enables more consumers to punish the seller once the seller shirks, it weakens their incentives to punish the seller once they observe him shirking.\\


\newtheorem{Proposition}{\hskip\parindent\bf{Proposition}}
\newtheorem{Theorem}{\hskip\parindent\bf{Theorem}}
\newtheorem{Lemma}{\hskip\parindent\bf{Lemma}}[section]
\newtheorem*{Lemma1}{\hskip\parindent\bf{No-Back-Loop Lemma}}
\newtheorem*{Lemma3}{\hskip\parindent\bf{No-Back-Loop Lemma*}}
\newtheorem*{Lemma2}{\hskip\parindent\bf{Learning Lemma}}
\newtheorem{Corollary}{\hskip\parindent\bf{Corollary}}
\newtheorem{Definition}{\hskip\parindent\bf{Definition}}
\newtheorem{Assumption}{\hskip\parindent\bf{Assumption}}
\newtheorem{Condition}{\hskip\parindent\bf{Condition}}
\newtheorem{Claim}{\hskip\parindent\bf{Claim}}
\newtheorem*{Assumption1}{\hskip\parindent\bf{Assumption 1'}}

\begin{spacing}{1.5}

\section{Introduction}\label{sec1}
Economic agents benefit from good reputations. This idea was formalized in the reputation literature pioneered by Fudenberg and Levine (1989), who show that a patient player (e.g., a seller) can secure a high payoff if he builds a reputation in front of a sequence of short-run players (e.g., consumers). These reputation results assume that the consumers can observe the \textit{full history} of the seller's actions, or can observe
\textit{a sufficiently long history} of the seller's actions including the \textit{exact sequence} of these actions. The intuition is that when the consumers observe that the seller has taken a particular action (e.g., exerting high effort) for a long time,  they will be convinced that he is likely to take that same action in the future.

In practice, consumers may not have long memories and may not know the exact sequence of the seller's actions. For example,
eBay and eLance only disclose the number of positive and negative ratings each seller received in the last 6 months but not the exact timing of these ratings (Dellarocas 2006). A firm's ranking in the Better Business Bureau is based on some aggregate statistics of its performance in the last 36 months but not on the other details. This is also the case in markets without good record-keeping institutions where consumers primarily learn by talking to other consumers: It is hard for the consumers to learn about the seller's actions in the distant past since people who bought from the seller a long time ago may forget  their experiences. Although people who bought recently may share their experiences with future consumers, they usually do not share all the details they know, such as who bought before them and what they learnt from others. This makes it hard for the consumers to learn the exact sequence of actions.

This paper takes a first step to analyze reputation effects when the consumers have limited memories and do not have detailed information about the seller's history. I study a novel reputation model in which every consumer can only observe some summary statistics of the seller's last $K$ actions but \textit{cannot observe the exact sequence of these actions}.\footnote{My reputation result extends when the short-run players can observe the summary statistics of the patient player's entire history, i.e., they observe the number of times that the patient player took each of his actions since the beginning of the game. In Section \ref{sub5.1}, I study an alternative model where the short-run players can observe the \textit{exact sequence} of the patient player's last $K$ actions.} This stands in contrast to the canonical reputation model of Fudenberg and Levine (1989) and existing reputation models with limited memories such as Liu and Skrzypacz (2014) which assume that the consumers can \textit{perfectly} observe the \textit{exact sequence} of the seller's last $K$ actions.


When the consumers do not have long memories, it is unclear whether the seller can still benefit from building reputations. This is because the consumers may not be convinced that the seller will exert high effort in the future if they can only observe him exerting high effort in the last few periods.

Nevertheless, I show that in a natural class of games, the seller can secure high returns from building reputations \textit{regardless} of the consumers' memory length $K$. This stands in contrast to existing reputation results, which \textit{require} the consumers to have long enough memories. I also show that the consumers can obtain their first-best welfare in all equilibria \textit{if and only if} their memories are short enough. This stands in contrast to
existing reputation models where the consumers can observe the exact sequence of actions (e.g., Fudenberg and Levine 1989), in which
there are equilibria where
the consumers receive low payoffs.

I study an infinitely repeated game between a patient player and a sequence of short-run players. Players' payoffs are \textit{monotone-supermodular} (MSM) in the sense that there exists a complete order on each player's action set
such that (i) the patient player's payoff decreases in his own action and increases in his opponents' action, and (ii) both players' payoff functions have strictly increasing differences.\footnote{The seller's payoff being monotone and the consumers' payoff being supermodular are standard assumptions in the reputation literature, which are also assumed in Mailath and Samuelson (2001), Liu (2011), Ekmekci (2011), Liu and Skrzypacz (2014), among many others. The difference is that Liu (2011) and Liu and Skrzypacz (2014) assume that the seller's payoff is \textit{submodular}. I motivate my supermodularity assumption in Section \ref{sub2.1}. I study the case where the seller has submodular payoff in Section \ref{sub5.2}.}
A leading example that satisfies my MSM condition is the product choice game in Mailath and Samuelson (2001, 2015):
\begin{center}
\begin{tabular}{| c | c | c |}
  \hline
   seller $\backslash$  consumer & {\color{blue}{\textbf{T}rust}} & \textbf{N}o Trust \\
  \hline
  {\color{blue}{\textbf{H}igh Effort}} & ${\color{blue}{\mathbf{1}}}, 1$ & $-c_N, x$ \\
  \hline
  \textbf{L}ow Effort & $1+c_T, -x$ & $0,0$ \\
  \hline
\end{tabular} with $0<c_T<c_N$ and $x \in (0,1)$.
\end{center}
This game satisfies MSM once players' actions are ranked according to $H \succ L$ and $T \succ N$.

The patient player privately observes his \textit{type}: He is either a \textit{commitment type} who chooses his highest action (in the example, high effort) in every period, or a \textit{strategic type} who maximizes his discounted average payoff. The patient player's \textit{reputation} is the probability his opponents assign to the commitment type.

For simplicity, my baseline model assumes that each short-run player can only observe the number of times that the patient player took each of his actions in the last  $K \in \mathbb{N}$ periods but not the exact sequence of these $K$ actions.\footnote{My theorems are robust when there is a small amount of noise in consumers' signals. My results are also robust when a small fraction of consumers know the exact sequence of the seller's last $K$ actions.
I also study an extension where there is a partition of the seller's action space such that the consumers can only observe which partition element the seller's last $K$ actions belong to.}
In order to be consistent with the existing literature on reputation games with limited memories, such as Liu (2011), Liu and Skrzypacz (2014), and Levine (2021), I make a standard assumption that the short-run players \textit{cannot} directly observe how long the game has lasted (i.e., calendar time). They have a prior belief about calendar time and update their beliefs via Bayes rule after observing their histories.


Theorem \ref{Theorem1} shows that as long as $K$ is at least $1$, the patient player receives at least his commitment payoff in every Nash equilibrium, and that he can secure this payoff by taking the highest action in  every period.
In the product choice game, my theorem implies that the patient seller's payoff is at least $1$ in every equilibrium. This conclusion stands in contrast to the repeated complete information game \textit{without} any commitment type, in which there are equilibria where the seller receives his minmax payoff $0$.

To the best of my knowledge, Theorem \ref{Theorem1} is the first reputation result that allows for \textit{arbitrary memory length}, and in particular, it allows the short-run players to have \textit{arbitrarily short memories}. This aspect of my result stands in contrast to the existing reputation results which \textit{assume} that the short-run players have infinite memories
(e.g., Fudenberg and Levine 1989) or  long enough memories (e.g., Theorem 2 in Liu and  Skrzypacz 2014).\footnote{Theorem 2 in Liu and Skrzypacz (2014) shows that the patient player can  secure his commitment payoff in all stationary equilibria when the short-run players' memory length $K$ is greater than some cutoff $\widehat{K}$, where $\widehat{K}$ depends on the prior probability of the commitment type. That is to say, their reputation result requires the short-run players to have \textit{long enough memories}.}
My result contributes to the reputation literature by showing that in a natural class of games, the patient player can secure high returns from building reputations even when his opponents do not have long memories and can only observe some summary statistics about his recent actions.

The challenge to prove this result comes from the observation that the short-run players may have arbitrarily short memories and cannot observe everything their predecessors observe. As a result, the standard arguments in Fudenberg and Levine (1989,1992), Sorin (1999), and Gossner (2011) do not apply.

My proof circumvents this challenge by establishing a \textit{no-back-loop property}, that it is never optimal for the patient player to \textit{milk his reputation when it is strictly positive and later restore his reputation}. I explain the intuition using the product choice game. Since the seller's payoff increases in consumer's trust but decreases in his effort, he has an incentive to restore his reputation only if the consumers trust him with higher probability after he restores his reputation. Since the seller's payoff is supermodular, he has a stronger incentive to exert high effort when the consumers trust him with higher probability. Hence, if it is optimal for the seller to restore his reputation when the consumers trust him with lower probability,
then it is \textit{not} optimal for him to milk his reputation when the consumers trust him with higher probability.\footnote{This argument is incomplete since it does not take into account the fact that the seller's incentive in each period depends not only on his stage-game payoff, but also on his continuation value. I present the complete argument in Section \ref{sub3.2}.}


Since the seller's equilibrium strategy satisfies the \textit{no-back-loop property}, there is \textit{at most one period} over the infinite horizon where he has exerted high effort in all of the last $K$ periods but will shirk in the current period. Since the commitment type exerts high effort in every period, the consumers believe that the seller will exert high effort with probability close to $1$
after they observe him exerting high effort in all of the last $K$ periods, causing them to have a strict incentive to trust the seller.
This implies Theorem \ref{Theorem1}, since in any equilibrium, the patient seller obtains his commitment payoff $1$ if he deviates and exerts high effort in every period, and his equilibrium payoff must be weakly greater than his payoff under any deviation.

Next, I examine consumer welfare. This is not covered by Theorem \ref{Theorem1}, which only shows that the seller receives at least his commitment payoff if he exerts high effort in every period. However, there might be other strategies that can give the seller  weakly higher payoffs. As a result,
it is unclear whether the seller will exert high effort in equilibrium and  whether the consumers will attain a high welfare.\footnote{For example, in the canonical reputation model of Fudenberg and Levine (1989), the patient player can secure his commitment payoff by taking his commitment action in every period. However, there are also equilibria where he takes the commitment action with low frequency and the short-run players receive their minmax payoff. See Li and Pei (2021) for details.}

My next set of results establish an equivalence between (i) the short-run player's memory length $K$ is below some cutoff,
(ii) the patient player taking his highest action in almost all periods in \textit{all} equilibria, and (iii) the short-run players approximately obtaining their first-best welfare in \textit{all} equilibria. I also show that when $K$ is below the cutoff, the patient player will take his highest action with probability close to one except for the initial few periods and periods in the distant future that have negligible payoff consequences.

I explain the intuition using the product choice game.
An increase in $K$ has two effects on the seller's incentives. First, once the seller chooses $L$, more consumers can observe it when $K$ is larger, in which case more consumers \textit{have the ability to punish} the seller.
However, a larger $K$ also makes it more difficult to \textit{motivate} the consumers to punish the seller. To see this, suppose the consumers believe that the strategic-type seller will play $L$ in periods $K-1,2K-1,...$ and will play $H$ in other periods. If a consumer observes that the seller played $L$ once in the last $K$ periods, then she knows that the seller is not the commitment type. According to Bayes rule, she believes that the seller will play $L$ in the current period with probability close to $\frac{1}{K}$. When $K$ is large, $\frac{1}{K}$ is small, so a consumer who does not observe calendar time believes that it is unlikely that the seller will play $L$ in the current period, and thus has no incentive to play $N$ even though she knows that the seller is not the commitment type. If the consumers play $T$ when $L$ occurred only once, then the seller prefers \textit{playing $L$ once every $K$ periods} to \textit{playing $H$ in every period}, making the consumers' beliefs self-fulfilling. This leads to an equilibrium where the seller exerts low effort periodically.\footnote{Using similar ideas, one can construct equilibria where the seller exerts low effort $n$ times every $K$ periods, provided that the consumers have no incentive to play $N$ when they believe that the seller will exert low effort with probability no more than $\frac{n}{K}$.}

When $K$ is small, I show that in \textit{every} equilibrium, the patient player's payoff is bounded below his commitment payoff
after he loses his reputation. Since Theorem \ref{Theorem1} implies that the patient player can secure his commitment payoff by taking the highest action in every period, he will do so in \textit{all} equilibria.

My proof introduces new techniques that can characterize the common properties of the patient player's behavior in \textit{all} equilibria. Some of my arguments require \textit{no assumption} on players' payoffs, which are portable to other repeated games with limited records, repeated games where players observe random samples of their opponents' past actions, and repeated games where players use finite automaton strategies.

I review the related literature in the rest of this section. I present my baseline model in Section \ref{sec2}. I state my main results in Section \ref{sec3}. I study several extensions in Section \ref{sec4}, such as the case where the consumers only observe noisy signals about the seller's last $K$ actions, and the case where there is a partition of the seller's action space such that the consumers only observe the number of times that the seller's last $K$ actions belong to each partition element. Section \ref{sec5} examines two alternative models that are only one-step-away from my baseline model and the one in Liu and Skrzypacz (2014): one in which the seller's payoff is supermodular but the consumers observe the exact sequence of the seller's last $K$ actions, and another one in which the consumers do not know the exact sequence of actions but the seller's payoff is submodular.

\paragraph{Related Literature:}
My paper contributes to three strands of literature: reputation with limited memories, cooperation under limited information, and the sustainability of reputations.

In contrast to the existing reputation models with limited memories such as Liu (2011), Liu and Skrzypacz (2014), and Pei (2022),\footnote{Pei (2022) assumes that the short-run players can observe the exact sequence of the patient player's last $K$ actions and at least one previous short-run player's action. He constructs an equilibrium where the patient player receives his minmax payoff.}
I introduce a novel reputation model in which the short-run players do not know the exact sequence of actions.\footnote{Levine (2021) assumes that the short-run players have 1-period memory, i.e., $K=1$, in which case whether the short-run players can observe the exact sequence of the last $K$ actions is irrelevant. In Jehiel and Samuelson (2012), the short-run players mistakenly believe that all types of the patient player use stationary strategies. Although the short-run players can observe the exact sequence of actions, their belief about the patient player's current-period action depends only on the empirical action frequencies.} Compared to
the reputation results in those papers which require long enough memories,
I show that the patient player can secure high returns from building reputations regardless of his opponents' memory length. My results also shed light on the effects of memory length on consumer welfare which, to the best of my knowledge, has not been examined in the existing reputation literature.\footnote{Kaya and Roy (2022) study a repeated signaling game where the consumers' best reply depends only on the seller's type. They show that longer memories \textit{encourage} the low-quality seller to imitate the high-quality one. In my model, the consumers' payoff depends only on players' actions and longer memories \textit{undermine} the seller's incentive to imitate the commitment type.}

My paper is also related to the literature on sustaining cooperation when players have limited information about others' past behaviors. This has been studied in repeated games with random matching by Kandori (1992), Ellison (1994), Takahashi (2010), Heller and Mohlin (2018), and Clark, Fudenberg and Wolitzky (2021). Most of these papers focus on the prisoner's dilemma where all players are patient.\footnote{For games with general payoffs, see Deb (2020), Deb, Sugaya and Wolitzky (2020), and Sugaya and Wolitzky (2020).}
Their results provide conditions on the monitoring technology under which \textit{either} a folk theorem holds \textit{or} players obtain their minmax payoffs in all equilibria. In contrast, I study repeated games between a patient player and a sequence of short-run players with one-sided lack of commitment (e.g., product choice games) instead of the prisoner's dilemma. I provide conditions under which players can secure high payoffs in \textit{all} equilibria.


Bhaskar and Thomas (2019) study a repeated \textit{complete information} game between a patient player and a sequence of short-run players.
They assume that the short-run players do not have any information about actions that were taken more than $K$ periods ago. They find information structures  under which players can cooperate in some equilibria. By contrast, I study a repeated \textit{incomplete information} game. I provide conditions under which the consumers approximately attain their first-best payoff in \textit{all} equilibria.

Ekmekci (2011) and Vong (2022) study repeated product choice games where the seller's cost is independent of the consumers' actions. They construct rating systems under which there \textit{exists} an equilibrium where the patient seller exerts high effort in almost all periods.
Although I do not explicitly study an information design problem, my results imply that when the seller has supermodular payoffs, he will exert high effort in almost all periods \textit{in all equilibria} under a simple disclosure rule, which is to reveal the number of times that the seller took each of his actions in the last $K$ periods for some small but positive $K$.

My results also contribute to the literature on reputation sustainability. Theorem \ref{Theorem2} focuses on a novel notion of reputation sustainability, namely, whether the patient player will take the commitment action with discounted frequency close to $1$. This notion of reputation sustainability is novel relative to the one in Cripps, Mailath and Samuelson (2004) which focuses on the patient player's behavior and reputation as $t \rightarrow +\infty$.\footnote{Ekmekci, Gossner and Wilson (2012) and Liu and Skrzypacz (2014) propose another notion of reputation sustainability, that the patient player can secure his commitment payoff \textit{at every history} in every equilibrium, rather than just securing his commitment payoff in period $0$. Theorem \ref{Theorem1} implies that in my model, reputation is sustainable under their criteria for any $K \geq 1$.}
 My notion is better suited for evaluating consumer welfare. For example, if reputation is sustainable under my notion, then the consumers can approximately attain their first-best welfare in all equilibria.

Pei (2020) and Ekmekci and Maestri (2022) study reputation models with interdependent values. They provide conditions under which the patient player takes his commitment action in almost all periods in all equilibria. Nevertheless, the mechanism behind their results is different from that behind mine.
The patient player is guaranteed to be punished in their interdependent value settings since deviating from the commitment action is a negative signal about the patient player's type. By contrast, the current paper studies a private-value model but the short-run players do not know the exact sequence of the patient player's actions. Since the short-run players cannot fine-tune their strategies based on the details of the game's history, the punishments needed to sustain cooperation inevitably punish the patient player at other histories, and harsher punishments at a larger set of histories can deter the patient player from milking his reputation.

\section{Baseline Model}\label{sec2}
Time is indexed by $t=0,1...$. A long-lived player $1$ (e.g., seller) interacts with a different player $2$ (e.g., consumer) in each period.
After each period, the game ends with probability $1-\delta$ with $\delta \in (0,1)$, after which players' stage-game payoffs are zero.
In the baseline model, player 1 is indifferent between receiving one unit of utility in the current period and in the next period, so he discounts his future payoffs by $\delta$.

In period $t$, player 1 chooses $a_t \in A$ and player $2_t$ chooses $b_t \in B$ simultaneously from finite sets $A$ and $B$. Players' stage-game payoffs are $u_1(a_t,b_t)$ and $u_2(a_t,b_t)$. All my results in Sections \ref{sec3} and \ref{sec4} are shown under the
following \textit{monotone-supermodularity} assumption on players' stage-game payoffs:
\begin{Assumption}\label{Ass1}
There exist a complete order $\succ_A$ on $A$ and a complete order $\succ_B$ on $B$ such that first, $u_1(a,b)$ is strictly increasing in $b$ and is strictly decreasing in $a$, and second, both $u_1(a,b)$ and $u_2(a,b)$ have strictly increasing differences in $a$ and $b$.
\end{Assumption}
The product choice game satisfies Assumption \ref{Ass1} under the rankings
 $H \succ_A L$ and $T \succ_B N$, where the requirements translate into (i) high effort is costly for the seller, (ii) the seller benefits from the consumers' trust, (iii) the consumers have stronger incentives to trust when effort is higher, and (iv) the cost of effort is lower when the consumers choose $T$.
The first three conditions are standard: They are satisfied in most applications and most games analyzed in the reputation literature, including those in Mailath and Samuelson (2001), Ekmekci (2011), Liu (2011), and Liu and Skrzypacz (2014).
The assumption that the seller's payoff being supermodular stands in contrast to some of the existing papers. For example, Ekmekci (2011) assumes that the seller's cost is independent of the consumers' actions, while Liu (2011) and Liu and Skrzypacz (2014) assume that the seller's cost is higher when the consumers trust him. I motivate my supermodularity assumption in Section \ref{sub2.1} and discuss the case where $u_1$ is weakly submodular in Section \ref{sub5.2}.

To highlight the mechanisms at work,
my baseline model focuses on games where player $2$'s action choice is binary, i.e., $|B|=2$. This class of games is a primary focus of the reputation literature, including Mailath and Samuelson (2001), Ekmekci (2011), Liu (2011), and Levine (2021). Unlike those papers that focus on $2 \times 2$ games, my baseline model allows the patient player to have any finite number of actions.

Section \ref{sub4.2} extends my theorems to games where $|B| \geq 3$, which include but are not limited to the ones in Liu and Skrzypacz (2014) where player $2$ has a unique best reply to every $\alpha \in \Delta (A)$.
My results also hold when $A$ is a lattice instead of a completely ordered set. This fits applications where player 1's action set is \textit{multi-dimensional} such as a seller choosing the quality of both his product and his customer service.

Before choosing $a_t$, player 1 observes all the past actions $h^t \equiv \{a_s,b_s\}_{s=0}^{t-1}$ and his perfectly persistent type $\omega \in \{\omega_s, \omega_c\}$. Let $\omega_c$ stand for a \textit{commitment type} who plays his highest action $a^* \equiv \max A$, or his \textit{commitment action}, in every period. Let $\omega_s$ stand for a \textit{strategic type} who maximizes his discounted average payoff $\sum_{t=0}^{\infty} (1-\delta) \delta^t u_1(a_t,b_t)$.
Let $\pi_0 \in (0,1)$ be the prior probability of the commitment type. Let $\pi_t$ be the probability that player $2_t$'s belief assigns to the commitment type, which I call player $1$'s \textit{reputation}.

Before choosing $b_t$, player $2_t$ only observes the number of times that player 1 took each of his actions in the last $\min \{t,K\}$periods, where $K \in \{1,2,...\}$ is a parameter that measures the society's memory length. An implication is that player 2 does not know the order with which player 1 took his last $K$ actions. For example, if $K=2$, then player 2 cannot distinguish between $(a_{1},a_{2})=(a^*,a')$ and $(a_{1},a_{2})=(a',a^*)$.

This is the key modeling innovation relative to existing reputation models with limited memories such as Liu (2011), Liu and Skrzypacz (2014) and Pei (2022), all of which assume that the short-run players can \textit{perfectly} observe the patient player's last $K$ actions as well as the \textit{exact sequence} of these $K$ actions.

I also make a standard assumption in repeated games with limited memories that player 2 \textit{cannot} directly observe how long the game has lasted. As in Liu and Skrzypacz (2014), I assume that player 2 has a full support prior belief about calendar time and update their beliefs via Bayes rule after observing their histories. The standard interpretation for this assumption is that due to consumers' limited memories, they cannot directly observe how long the seller has been in the market.
Under this formulation, the first $K$ short-run players observe fewer than $K$ actions, so their \textit{posterior beliefs} assign probability $1$ to the true calendar time.

What is a reasonable prior belief about calendar time? Recall that the game ends with probability $1-\delta$ after each period. Therefore, for every $t \in \{0,1,...\}$, the probability player 2's prior assigns to calendar time being $t+1$ should equal $\delta$ times the probability her prior assigns to calendar time being $t$. The unique prior that satisfies this for every $t$ is the one that assigns probability $(1-\delta)\delta^t$ to calendar time being $t$.


The set of player $1$'s histories is
$\mathcal{H}_1 \equiv \{
    (a_s,b_s)_{s=0}^{t-1} \textrm{ s.t. } t \in \mathbb{N} \textrm{ and } (a_s,b_s) \in A \times B
    \}$ with a typical element $h^t$. The set of player $2$'s histories is
    $\mathcal{H}_2 \equiv \{
    (n_1,...,n_{|A|}) \in \mathbb{N}^{|A|} \textrm{ s.t. } n_1 \geq 0,...,n_{|A|} \geq 0 \textrm{ and } n_1+...+n_{|A|} \leq K
    \}$, where $n_1,...,n_{|A|}$ are the number of times that player $1$ played each of his actions in the last $K$ periods.
A typical element of $\mathcal{H}_2$ is $h_2^t$.
Strategic-type player 1's strategy is $\sigma_1 : \mathcal{H}_1 \rightarrow \Delta (A)$.
Player $2$'s strategy is $\sigma_{2} : \mathcal{H}_{2} \rightarrow \Delta (B)$. Let $\Sigma_i$ be the set of player $i$'s strategies.
Under my formulation, player $2$'s action depends only on the history she observes. For example, player $2_t$ and player $2_{t+1}$ will take the same (possibly mixed) action if they observe the same history. This is a standard requirement in reputation models with limited memories, such as Liu (2011), Liu and Skrzypacz (2014) and Levine (2021).

\subsection{Discussion of Modeling Choices and Extensions}\label{sub2.1}
The consumers in my model do not know the exact sequence of the seller's actions.
This is motivated by situations such as (i) online platforms such as eBay and eLance that only disclose the number of positive and negative ratings each seller received in the last 6 months, (ii) rating institutions such as the Better Business Bureau whose ranking depends only on some aggregate
 statistics of a firm's performance in the last 36 months, and (iii) markets without good record-keeping institutions, such as informal markets in developing countries, where consumers cannot easily obtain detailed information about the exact sequence of the seller's behaviors.
My results hold when a small fraction $\varepsilon$ of consumers know the exact sequence of actions. I discuss the case where the consumers can perfectly observe the exact sequence in Section \ref{sub5.1}.


My baseline model rules out imperfect monitoring by assuming that each consumer knows the number of times that the seller took each of his actions in the last $K$ periods. Due to the technical challenges in analyzing repeated \textit{incomplete information} games with \textit{bounded memories},
most of the existing papers in this literature including Liu (2011) and Liu and Skrzypacz (2014)
 also rule out imperfect monitoring by assuming that player $2$ can \textit{perfectly} observe player $1$'s last $K$ actions.\footnote{Although Bhaskar and Thomas (2019) allow for imperfect monitoring, there is \textit{no incomplete information} in their model.}

I study two extensions that allow for \textit{imperfect monitoring}, which are motivated by situations such as the consumers' experiences (or the ratings) are noisy signals of the seller's effort, or the consumers only communicate coarse information about the seller's action to future consumers. In Section \ref{sub4.4}, I assume that there is some noisy signal $\widetilde{a}_t$ about $a_t$ and player $2_t$ observes the number of times that each signal realization occurred in the last $\min \{t,K\}$ periods. In Section \ref{sub4.3}, I assume that the short-run players can only learn from \textit{coarse summary statistics}, which is characterized by a partition of $A$ such that player $2_t$ only observes the number of times that player $1$'s last $\min \{t,K\}$ actions belong to each partition  element.
My baseline model corresponds to the finest partition of $A$. Player 2 learns nothing under the coarsest partition of $A$.

My baseline model assumes that the short-run players cannot directly observe calendar time, which is also assumed in Liu and Skrzypacz (2014), Section 2 in Acemoglu and Wolitzky (2014), Cripps and Thomas (2019), and Levine (2021). My baseline model focuses on an exponential prior belief about calendar time, which is natural when $\delta$ is interpreted as the probability with which the game continues after each period. Hu (2020) provides a microfoundation for this prior belief by constructing a model with random entry order.\footnote{Liu (2011) and Heller and Mohlin (2018) focus on \textit{stationary equilibria} where strategies are \textit{required} to be time-independent, which according to Liu (2011), is equivalent to having an improper uniform prior about calendar time. My results extend to Liu (2011)'s setting where player 1's discount factor is less than $1$ and player 2 has an improper uniform prior about calendar time.}

Section \ref{sub4.1} extends my theorems to a model where player $1$'s discount factor is different from the game's continuation probability. This is the case when player $1$ exits the game with positive probability after each period and values his utility in the current period more than his future utilities.  As will become clear in my proofs, my theorems extend to all prior beliefs as long as (i) the probability of any calendar time is close to $0$, and (ii) the ratio between the probabilities of any two adjacent calendar times is close to $1$.

My assumptions on players' payoffs are standard except for $u_1(a,b)$ having strictly increasing differences. In the product choice game, my assumption requires that the seller's effort and the consumers'
trust to be strategic complements. This assumption fits, for example, when each consumer chooses between buying a \textit{high-end version} (action $T$) and a \textit{low-end version} (action $N$) of a product and the seller decides whether to provide good customer service, which requires him to exert high effort. It is reasonable to assume that the seller's cost of providing good service is lower when the consumers buy the high-end version, since the high-end version breaks down less frequently compared to the low-end version.

My assumption also fits when the seller's \textit{monetary cost} of supplying high quality is independent of the consumers' action but players are altruistic and internalize a fraction of other players' monetary payoffs. Evidence for altruism has been widely documented by psychologists, see for example Batson and Shaw (1991). Levine (1998) formalizes altruism using a model where each player maximizes a convex combination of \textit{his own monetary payoff} and \textit{others' monetary payoffs}. In my setting, when the seller internalizes a positive fraction of the consumers' monetary payoffs, his real cost of effort is strictly lower when the consumers choose $T$ since the consumers benefit more from the seller's effort when they choose $T$.

My baseline model assumes that there is only one commitment type who plays a stationary pure strategy. This is also assumed in most of the existing reputation models with limited memories such as Liu (2011) and Liu and Skrzypacz (2014). My theorems are robust when there are multiple stationary pure-strategy commitment types, as long as the type who plays $a^*$ in every period occurs with positive probability.

\section{Results}\label{sec3}
Section \ref{sub3.2} establishes a reputation result that allows for \textit{arbitrary memory length}, that a sufficiently patient player 1 receives at least his commitment payoff in all equilibria as long as $K \geq 1$. My proof uses a \textit{no-back-loop property} that applies to all of  player 1's best replies, regardless of his discount factor.
Section \ref{sub3.3} shows that player 2 can approximately attain their highest feasible payoff and the patient player will play $a^*$ in almost all periods in all equilibria \textit{if and only if} $K$ is below some cutoff. Some of the arguments in my proof, such as Lemmas \ref{L4.1}, \ref{Linflow}, and \ref{L4.3}, apply \textit{independently} of players' payoffs and incentives. They are portable to other repeated games with limited memories, repeated games where players observe a random sample of their opponents' past actions, and repeated games where players use finite automaton strategies.

\subsection{Reputation Result for Arbitrary Memory Length}\label{sub3.2}
My first result shows that for every $K \geq 1$ and $\pi_0$, player 1 can approximately secure his commitment payoff in \textit{all} Nash equilibria as $\delta \rightarrow 1$,
and that he can secure this payoff by playing $a^*$ in every period.

Formally,
let $b^*$ be player $2$'s lowest best reply to $a^*$. Following Fudenberg and Levine (1989), I call $u_1(a^*,b^*)$ player $1$'s \textit{commitment payoff}. Let $\underline{a} \equiv \min A$ be player 1's lowest action.
Let $\underline{b}$ be player $2$'s lowest best reply to $\underline{a}$.
For every $\pi_0 \in (0,1)$, there exists  $\underline{\delta}(\pi_0) \in (0,1)$ such that for every $\delta > \underline{\delta}(\pi_0)$, each of player 2's best reply to  mixed action
     $\big\{1-\frac{(1-\delta)(1-\pi_0)}{\pi_0}
     \big\} a^* +  \frac{(1-\delta)(1-\pi_0)}{\pi_0} \underline{a}$
is no less than $b^*$. Such $\underline{\delta}(\pi_0) \in (0,1)$ exists for every $\pi_0 \in (0,1)$ since $b^*$ is defined as the lowest best reply to $a^*$, the value of $\frac{(1-\delta)(1-\pi_0)}{\pi_0}$ converges to $0$ as $\delta \rightarrow 1$, and best reply correspondences are upper-hemi-continuous.
\begin{Theorem}\label{Theorem1}
Suppose $\delta > \underline{\delta} (\pi_0)$ and $K \geq 1$. Player $1$'s payoff in any Nash equilibrium is at least
\begin{equation}\label{commitmentpayoffbound}
(1-\delta^K) u_1(a^*,\underline{b})+\delta^K u_1(a^*,b^*).
\end{equation}
\end{Theorem}
The payoff lower bound (\ref{commitmentpayoffbound}) converges to $u_1(a^*,b^*)$ as $\delta \rightarrow 1$.  Therefore,
Theorem \ref{Theorem1} identifies a natural class of games such that regardless of the short-run players' memory length $K$,\footnote{When $K=+\infty$, one can use Fudenberg and Levine (1989)'s argument to show that player $1$ can secure payoff approximately $u_1(a^*,b^*)$ in every Nash equilibrium as $\delta \rightarrow 1$.
My proof focuses on the case where $K$ is finite, which requires new arguments.}
a sufficiently patient player can secure at least his commitment payoff $u_1(a^*,b^*)$ by building a reputation for playing $a^*$.
As will become clear later in the proof, at every history of every Nash equilibrium, if player $1$ deviates by playing $a^*$ in every subsequent period, then his continuation value  after $K$ periods is at least $u_1(a^*,b^*)$.
By contrast, in the repeated complete information game \textit{without} any commitment type, there are equilibria where player $1$ plays $\underline{a}$ and player $2$ plays $\underline{b}$ in every period and player $1$ receives his minmax payoff $u_1(\underline{a},\underline{b})$.

To the best of my knowledge, Theorem \ref{Theorem1} is the first reputation result that allows the short-run players to have arbitrary memory length, and in particular, they may have \textit{arbitrarily short memories}.
This aspect of my result stands in contrast to existing reputation results which require the short-run players to have
\textit{infinite memories} (e.g., Fudenberg and Levine 1989), or \textit{long enough memories}
(e.g., Theorem 2 in Liu and Skrzypacz 2014),\footnote{Theorem 2 in Liu and Skrzypacz (2014) shows that for every $\pi_0>0$, there exists  $\widehat{K} \in \mathbb{N}$ such that a patient player 1 can approximately secure his commitment payoff in every equilibrium when $K>\widehat{K}$, i.e., $K$ needs to be large enough.} or \textit{infinite memories} about some noisy signal that can statistically identify the patient player's action
 (e.g., Fudenberg and Levine 1992, Gossner 2011, Theorem 2 in Pei 2022).


Since the short-run players have limited memories and cannot observe everything their predecessors observe, their belief is not a martingale process and the standard techniques in
Fudenberg and Levine (1989, 1992), Sorin (1999), and Gossner (2011) do not apply.
To overcome these challenges, my proof uses an observation called the \textit{no-back-loop property}, that
it is \textit{never} optimal for the patient player to milk his reputation when it is strictly positive and later restore his reputation.
I state this observation as a lemma, provide a heuristic explanation, and use this result to show Theorem \ref{Theorem1} by the end of this section.

Let
 $\mathcal{H}_1^* \equiv \big\{(a_s,b_s)_{s=0}^{t-1} \textrm{ } s.t. \textrm{ } t \geq K \textrm{ and } (a_{t-K},...,a_{t-1})=(a^*,...,a^*)\big\}$
be the set of player $1$'s histories where all of his last $K$ actions were $a^*$.
Let $\mathcal{H}_1(\sigma_1,\sigma_2)$ be the set of histories that occur with positive probability under $(\sigma_1,\sigma_2)$. Let $U_1(\sigma_1,\sigma_2)$ be player 1's discounted average payoff under $(\sigma_1,\sigma_2)$. Strategy $\widehat{\sigma}_1$ best replies to $\sigma_2$ if $\widehat{\sigma}_1 \in \arg\max_{\sigma_1 \in \Sigma_1} U_1(\sigma_1,\sigma_2)$, i.e., $\widehat{\sigma}_1$ maximizes player $1$'s payoff against $\sigma_2$.
\begin{Lemma1}
For every $\sigma_2: \mathcal{H}_2 \rightarrow \Delta (B)$ and pure strategy $\widehat{\sigma}_1: \mathcal{H}_1 \rightarrow A$ that best replies to $\sigma_2$, there exists \textit{no} $h^t \in \mathcal{H}_1(\widehat{\sigma}_1,\sigma_2) \bigcap \mathcal{H}_1^*$
such that when player $1$ uses strategy $\widehat{\sigma}_1$, he plays an action that is not $a^*$ at $h^t$ and reaches another
history that belongs to $\mathcal{H}_1(\widehat{\sigma}_1,\sigma_2) \bigcap \mathcal{H}_1^*$
in the future.
\end{Lemma1}

My \textit{no-back-loop lemma} rules out situations depicted in the left panel of Figure 1: Player $1$'s best reply $\widehat{\sigma}_1$ asks him to play $a' (\neq a^*)$ when his last $K$ actions were $a^*$ (the green circle) and then returns to a history where his last $K$ actions were $a^*$. That is to say, as soon as player $1$ milks his reputation at a history where his reputation is strictly positive, he will never return to any history where he has a positive reputation. This property applies to \textit{all} of player 1's best replies in the repeated game, not just to his equilibrium strategies. It also does not require player $1$ to be patient by allowing him to have any arbitrary discount factor $\delta \in (0,1)$.

The \textit{reputation cycles} ruled out by my lemma occur in all stationary equilibria in Liu and Skrzypacz (2014). This is driven by two modeling differences. First, they assume that $u_1$ has \textit{strictly decreasing differences} while I assume that $u_1$ has strictly increasing differences. Second, they assume that player $2$ can \textit{perfectly} observe the exact sequence of player 1's last $K$ actions while I assume that player 2 do not know the exact sequence of player 1's last $K$ actions. In Section \ref{sec5}, I study two alternative models that are only one-step-away from both my baseline model and the one in  Liu and Skrzypacz (2014).
\begin{figure}\label{Figure1}
\begin{center}
\begin{tikzpicture}[scale=0.33]
\node at (-8.3,-4) {$(a'',a^*,...,a^*)$};
\node at (8.3,-4) {$(a^*,...,a^*,a')$};
\node at (0,1.5) {$(a^*,...,a^*)$};
\node at (20,1.5) {$(a^*,...,a^*)$};
\node at (35,1.5) {$(a^*,...,a^*)$};
\node at (31,-2.5) {$(a'',a^*...,a^*)$};
\node at (39,-2.5) {$(a^*,...,a^*,a'')$};
\draw [fill=green, thick] (0,0) circle [radius=1];
\draw [fill=cyan, thick] (4,-4) circle [radius=1];
\draw [fill=yellow, thick] (4,-10) circle [radius=1];
\draw [fill=yellow, thick] (-4,-10) circle [radius=1];
\draw [fill=white, thick] (-4,-4) circle [radius=1];
\draw [ultra thick, ->] (0,0)--(1.5,-1.5)node[right]{$a'$}--(4,-4);
\draw [blue, ultra thick, ->] (4,-4)--(4,-10);
\draw [blue, ultra thick, ->, dotted] (4,-10)--(-4,-10);
\draw [blue, ultra thick, ->] (-4,-10)--(-4,-4);
\draw [ultra thick, ->] (-4,-4)--(-2,-2)node[left]{$a^*$}--(0,0);
\node [below] at (0,-12) {Player 1 uses strategy $\widehat{\sigma}_1$};
\node [below] at (0,-13.5) {that violates no-back-loop};
\node [below] at (20,-12) {Player 1 uses Deviation A};
\draw [fill=green, thick] (20,0) circle [radius=1];
\draw [fill=cyan, thick] (24,-4) circle [radius=1];
\draw [fill=yellow, thick] (24,-10) circle [radius=1];
\draw [fill=yellow, thick] (16,-10) circle [radius=1];
\draw [fill=white, thick] (16,-4) circle [radius=1];
\draw [blue, ultra thick, ->] (24,-4)--(24,-10);
\draw [blue, ultra thick, ->, dotted] (24,-10)--(16,-10);
\draw [blue, ultra thick, ->] (16,-10)--(16,-4);
\draw [ultra thick, ->] (16,-4)--(20,-4)node[above]{$a'$}--(24,-4);
\draw [fill=green, thick] (35,0) circle [radius=1];
\draw [fill=pink, thick] (39,-4) circle [radius=1];
\draw [fill=pink, thick] (39,-10) circle [radius=1];
\draw [fill=pink, thick] (31,-10) circle [radius=1];
\draw [fill=white, thick] (31,-4) circle [radius=1];
\draw [red, ultra thick, ->] (39,-4)--(39,-7)node[right]{$a^*$}--(39,-10);
\draw [red, ultra thick, ->, dotted] (39,-10)--(35,-10)node[below]{$a^*$}--(31,-10);
\draw [red, ultra thick, ->] (31,-10)--(31,-7)node[right]{$a^*$}--(31,-4);
\draw [ultra thick, ->] (31,-4)--(35,-4)node[below]{$a''$}--(39,-4);
\node [below] at (35,-12) {Player 1 uses Deviation B};
\end{tikzpicture}
\caption{The green circle represents a history that belongs to $\mathcal{H}_1^*$.
The white circle represents a history where player 1 is one-period-away from $\mathcal{H}_1^*$. The blue circle represents a history
that is reached after player 1 plays $a'$ at the green circle. The yellow circles represent histories that are reached when player 1 plays $\widehat{\sigma}_1$. The pink circles represent histories where
$a''$ occurred once and $a^*$ occurred $K-1$ times.}
\end{center}
\end{figure}
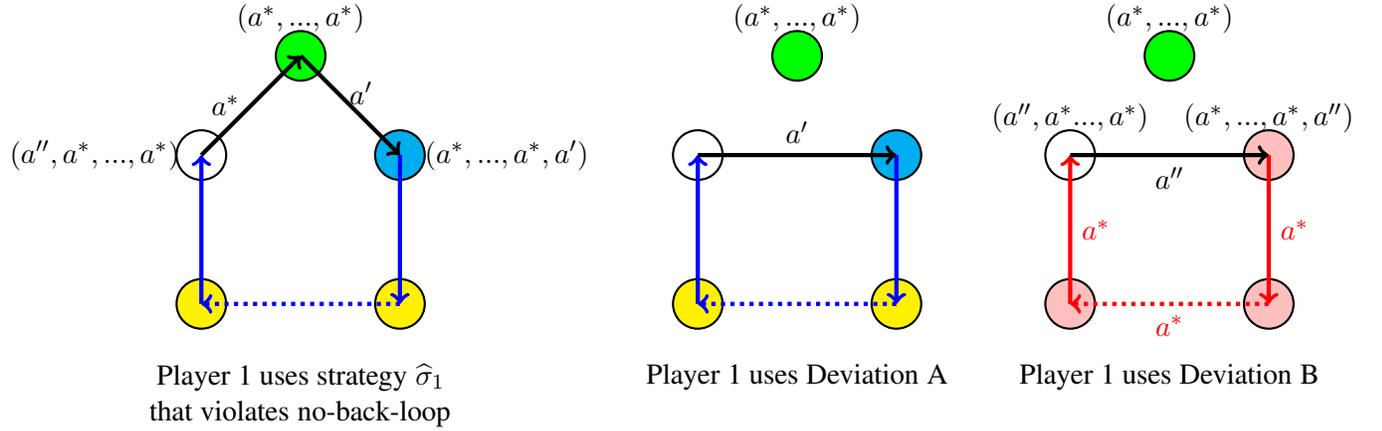

The proof of my lemma does \textit{not} follow from existing results on supermodular games, most of which focus on static games. This is because even when player 1's stage-game payoff $u_1(a,b)$ has strictly increasing differences, it is not necessarily the case that when the game is played \textit{repeatedly}, player 1 has a stronger incentive to play higher actions at histories where player 2's actions are higher. This is because player 1's current-period action affects future player 2's observations, which in turn affects future player 2's actions as well as player 1's continuation value.
I sketch a proof below. The detailed calculations are in Appendix \ref{secA}.

\begin{proof}[Proof Sketch:] Since player $2$ has no information about player $1$'s action more than $K$ periods ago, it is without loss to focus on player $1$'s pure-strategy best replies that depend only on his last $K$ actions, including the order of these $K$ actions.
Suppose by way of contradiction that there exists $\widehat{\sigma}_1$ that best replies to $\sigma_2$ such that  $\widehat{\sigma}_1$
plays $a'$ $(\neq a^*)$ at a history $h^t$ where $(a_{t-K},...,a_{t-1})=(a^*,...,a^*)$, and after a finite number of periods, reaches a history
$h^s$ that satisfies $(a_{s-K},...,a_{s-1})=(a'',a^*,...,a^*)$ where $a'' \neq a^*$, and then
plays $a^*$ at $h^s$ after which all of the last $K$ actions are $a^*$ again. Note that $a'$ and $a''$ can be the same.

I depict strategy $\widehat{\sigma}_1$ in the left panel of Figure 1, where
$h^t$ is represented by the green circle and
$h^s$ is represented by the white circle.
In what follows,
I propose \textit{two deviations} for player 1 starting from the white circle. I will show that \textit{at least one of them is strictly profitable}.
\begin{itemize}
  \item \textbf{Deviation A:} Plays $a'$ at the white circle, and then follows strategy $\widehat{\sigma}_1$.
  \item \textbf{Deviation B:} Plays $a''$ at the white circle, then plays $a^*$ for $K-1$ consecutive periods after which play will reach the white circle again, and then follows strategy $\widehat{\sigma}_1$.
\end{itemize}
These deviations are depicted in the middle and right panels of Figure 1.
I compare player 1's continuation value at the white circle
when he uses $\widehat{\sigma}_1$ to those under the two deviations:
\begin{enumerate}
  \item Compared to  $\widehat{\sigma}_1$, Deviation A takes a lower-cost action $a'$ at the white circle, skips the green circle, and frontloads the payoffs along the blue lines (i.e., the blue, yellow, and white circles). If player 1 prefers  $\widehat{\sigma}_1$ to Deviation A, then his average payoff from the circles along the blue lines (i.e., the payoff that Deviation A frontloads) must be strictly lower than his stage-game payoff at the green circle.
  \item Compared to $\widehat{\sigma}_1$, Deviation B
takes a lower-cost action $a''$ at the white circle, skips the green circle, and induces payoffs along the red lines in the next $K-1$ periods (i.e., the pink circles).
If player 1 prefers $\widehat{\sigma}_1$ to Deviation B, then his average payoff along the red lines must be strictly smaller than a convex combination of his payoff at the green circle and his average payoff along the blue lines.
\end{enumerate}
If $\widehat{\sigma}_1$ is player 1's best reply, then both Deviation A and Deviation B are unprofitable. Therefore, player 1's stage-game payoff at the green circle must be strictly greater than his average payoff along the red lines.

If $\widehat{\sigma}_1$ best replies to $\sigma_2$, then player $1$ prefers $a'$ to $a^*$ at the green circle and prefers $a^*$ to $a'$ at the white circle. No matter whether player 1 is currently at the green or the white circle, he will reach the green circle after playing $a^*$ and will reach the blue circle after playing $a'$. Hence,
the difference in player 1's incentives at the green and the white circles \textit{cannot} be driven by his continuation value. This implies that such a difference in incentives can only be driven by player 1's stage-game payoff, which is affected by player 2's actions at the green and the white circles.
Since $u_1(a,b)$ has strictly increasing differences and $a^* \succ_A a'$, the Topkis Theorem implies that it \textit{cannot} be the case that player 2's mixed action at the green circle strictly FOSDs her mixed action at the white circle. Therefore, player 2's action at the white circle weakly FOSDs her action at the green circle. Since player 2 cannot observe the order of player 1's last $K$ actions,
player 2's action at every circle along the red line coincides with her action at the white circle.

This leads to a contradiction since on the one hand, player $1$'s stage-game payoff at the green circle is strictly greater than his average payoff along the red lines, and on the other hand, player 2's action at the white circle weakly FOSDs her action at the green circle and player 1's stage-game payoff is strictly increasing in player 2's action. This contradiction implies that at the white circle,
 either Deviation A or Deviation B yields a strictly higher payoff for player 1 compared to strategy $\widehat{\sigma}_1$, so $\widehat{\sigma}_1$ is not a best reply.
 \end{proof}

I use the no-back-loop lemma to show Theorem \ref{Theorem1}.
\begin{proof}[Proof of Theorem 1:]
For every $t \in \mathbb{N}$, let $E_t$ be the event that \textit{player 1 is strategic and no action other than $a^*$ was played from period $\max\{0,t-K\}$ to period $t-1$}.
Fix any $\sigma_1 \in \Sigma_1^*$ and $\sigma_2$, let $p_t(\sigma_1,\sigma_2)$ be the \textit{ex ante} probability of event $E_t$ when the strategic-type player 1 plays $\sigma_1$ and player 2 plays $\sigma_2$. Since player 1 is the commitment type with probability $\pi_0$, we have $p_t (\sigma_1,\sigma_2) \leq 1-\pi_0$ for every  $t \in \mathbb{N}$.
Let $\mathbb{N}^* (\sigma_1,\sigma_2) \subset \mathbb{N}$ be the set of calendar time $t$ such that $p_t (\sigma_1,\sigma_2)>0$ and $t \geq K$.
For every $t \in \mathbb{N}^* (\sigma_1,\sigma_2)$, let $q_t (\sigma_1,\sigma_2)$ be the probability that player 1 \textit{does not} play $a^*$ in period $t$ conditional on event $E_t$.

Let $\Sigma_1^*$ be the set of player 1's pure strategies that satisfy the no-back-loop property.
For every $\sigma_1 \in \Sigma_1^*$, the definition of the no-back-loop property implies that $\sum_{t \in \mathbb{N}^*(\sigma_1,\sigma_2)} p_t (\sigma_1,\sigma_2) q_t (\sigma_1,\sigma_2)
\leq 1-\pi_0$.

Fix any arbitrary Nash equilibrium $(\widetilde{\sigma}_1,\sigma_2)$, the no-back-loop lemma implies that $\widetilde{\sigma}_1 \in \Delta (\Sigma_1^*)$.
For every pure strategy $\sigma_1 \in \Sigma_1^*$,
let $\widetilde{\sigma}_1 (\sigma_1)$ be the probability with which mixed strategy $\widetilde{\sigma}_1$ assigns to $\sigma_1$, which is well-defined since $\Sigma_1^*$ is a countable set.
Recall that player 2's prior belief assigns probability $(1-\delta)\delta^t$ to the calendar time being $t$. According to Bayes rule, at any history after period $K$ where all of player 1's last $K$ actions were $a^*$, player 2 believes that player 1's action is \textit{not} $a^*$ with probability
\begin{equation}\label{prob}
    \frac{ \sum_{\sigma_1 \in \Sigma_1^*} \widetilde{\sigma}_1(\sigma_1) \sum_{t \in \mathbb{N}^*(\sigma_1,\sigma_2)} (1-\delta) \delta^t p_t (\sigma_1,\sigma_2) q_t (\sigma_1,\sigma_2)}{\pi_0\sum_{t=K}^{+\infty}(1-\delta)\delta^t+\sum_{\sigma_1 \in \Sigma_1^*} \widetilde{\sigma}_1(\sigma_1)
    \sum_{t \in \mathbb{N}^*(\sigma_1,\sigma_2)} (1-\delta) \delta^t p_t(\sigma_1,\sigma_2)}.
\end{equation}
The denominator of (\ref{prob}) is at least $\pi_0 \delta^K$. Since $\sum_{t \in \mathbb{N}^*(\sigma_1,\sigma_2)} p_t (\sigma_1,\sigma_2) q_t (\sigma_1,\sigma_2) \leq 1-\pi_0$ for every $\sigma_1 \in \Sigma_1^*$, $\Sigma_1^*$ is a countable set,
and $t \geq K$ for every $t \in \mathbb{N}^*(\sigma_1,\sigma_2)$,
the numerator of (\ref{prob}) is no more than $(1-\delta)(1-\pi_0)\delta^K$. This suggests that (\ref{prob}) is no more than $\frac{(1-\delta)(1-\pi_0)}{\pi_0}$.
The definition of $\underline{\delta}(\pi_0)$ together with $u_2(a,b)$ having strictly increasing differences implies that when $\delta > \underline{\delta} (\pi_0)$,
actions strictly lower than $b^*$ are \textit{not} optimal for player $2$ when
 player $1$'s last $K$ actions were $a^*$. Moreover, since $\underline{b}$ is player $2$'s lowest best reply to player $1$'s lowest action $\underline{a}$, actions strictly lower than $\underline{b}$ is never optimal for player $2$.
 Hence, in any Nash equilibrium, if player 1 plays $a^*$ in every period, his discounted average payoff is at least $(1-\delta^K) u_1(a^*,\underline{b})+\delta^K u_1(a^*,b^*)$. This is a lower bound for player 1's equilibrium payoff.
\end{proof}

\subsection{Equilibrium Behavior \& Consumer Welfare}\label{sub3.3}
Although Theorem \ref{Theorem1} shows that player $1$ can secure his commitment payoff $u_1(a^*,b^*)$ \textit{if} he plays $a^*$ in every period, it does not imply that he will  play $a^*$ in equilibrium. This is because other strategies may give player $1$ weakly higher payoffs.
For example, Li and Pei (2021) show that in Fudenberg and Levine (1989)'s reputation model, although player $1$ can secure his commitment payoff by playing $a^*$ in every period, there are many equilibria where he plays $a^*$ with low frequency. Understanding whether player $1$ will actually play $a^*$ is important since his action frequencies affect consumer welfare.

My next set of results examine the patient player's equilibrium behavior and the short-run players' welfare.
Since the game continues with probability $\delta$ after each period, the sum of the short-run players' payoffs is $\mathbb{E}^{\sigma}\Big[
\sum_{t=0}^{+\infty} \delta^t u_2(a_t,b_t) \Big]$, where
$\mathbb{E}^{\sigma}[\cdot]$ denotes the expectation induced by $\sigma \equiv (\sigma_1,\sigma_2)$. Let
\begin{equation}\label{4.2}
U_2^{\sigma} \equiv \mathbb{E}^{\sigma}\Big[
   \sum_{t=0}^{+\infty} (1-\delta)\delta^t u_2(a_t,b_t)
\Big]
\end{equation}
be the \textit{normalized sum} of player 2's payoffs, which I use to measure consumer welfare. By definition, $U_2^{\sigma}=\sum_{(a,b) \in A \times B} F^{\sigma}(a,b) u_2(a,b)$ where
\begin{equation}\label{4.1}
   F^{\sigma}(a,b) \equiv \mathbb{E}^{\sigma}\Big[
   \sum_{t=0}^{+\infty} (1-\delta)\delta^t \mathbf{1}\{a_t=a,b_t=b\}
  \Big]
\end{equation}
is the \textit{discounted frequency} (or the \textit{occupation measure}) of action profile $(a,b)$.
Hence, the sum of the short-run players' payoffs depends on $(\sigma_1,\sigma_2)$ only through the discounted frequencies $\{F^{\sigma}(a,b)\}_{(a,b) \in A \times B}$.
\begin{Theorem}\label{Theorem2}
There exists a cutoff $\overline{K} \in \mathbb{N}$ that depends only on $(u_1,u_2)$ such that:
\begin{enumerate}
\item There exists a constant $C \in \mathbb{R}_+$ that is independent of $\delta$ such that
for every $1 \leq K < \overline{K}$, we have
$\sum_{b \succeq b^*} F^{\sigma}(a^*,b) \geq 1-(1-\delta)C$ in every Nash equilibrium $\sigma$ under $K$ and $\delta$.
\item There exists $\eta>0$ such that for every $K \geq \overline{K}$,  there exists $\underline{\delta} \in (0,1)$ such that for every $\delta > \underline{\delta}$, there exists a Perfect Bayesian equilibrium (PBE) such that $\sum_{b \in B} F^{\sigma}(a^*,b) \leq 1-\eta$.\footnote{Theorem \ref{Theorem1} and Statement 1 of Theorem \ref{Theorem2} study the common properties of all equilibria, which are stronger when I use weaker solution concepts such as Nash equilibrium.
Statement 2 of Theorem \ref{Theorem2} is about the existence of a type of equilibria, which is stronger under stronger solution concepts such as PBE. In the proof of Statement 2 of Theorem \ref{Theorem2} and other constructive proofs, I construct PBEs that satisfy sequential rationality and \textit{no signaling what you don't know} in Fudenberg and Tirole (1991).}
\end{enumerate}
\end{Theorem}
Theorem \ref{Theorem2} implies that (i) player 1 plays $a^*$ and player 2's action is at least $b^*$ in almost all periods in all Nash equilibria when player 2's memory $K$ is lower than some cutoff, and (ii) there are PBEs where player 1 plays $a^*$ with frequency bounded below $1$ when $K$ is above the cutoff. Note that the presence of the commitment type is necessary for this result, since in a repeated complete information game without commitment type, there exist equilibria where players play $(\underline{a},\underline{b})$ in every period regardless of $\delta$ and $K$.

Under an additional mild assumption on the short-run players' payoff, that $u_2(a,b)$ is strictly increasing in $a$, interpreted as the consumers' payoff increases in the seller's effort, $u_2(a^*,b^*)$ is the short-run players' \textit{highest feasible payoff}.
This is because $u_2(a^*,b) \geq u_2(a,b)$ for every $(a,b) \in A \times B$ and $u_2(a^*,b^*) \geq u_2(a^*,b)$ for every $b \in B$.
Theorem \ref{Theorem3} is a direct implication of Theorem \ref{Theorem2}, which provides a necessary and sufficient condition under which the short-run players attain their highest feasible payoff in all equilibria.
\begin{Theorem}\label{Theorem3}
Suppose $(u_1,u_2)$ satisfies Assumption \ref{Ass1} and $u_2(a,b)$ is strictly increasing in $a$.
\begin{enumerate}
  \item There exists a constant $C_0 \in \mathbb{R}_+$ that is independent of $\delta$ such that
for every $1 \leq K < \overline{K}$, we have $U_2^{\sigma} \geq u_2(a^*,b^*)-C_0 (1-\delta)$ in every Nash equilibrium $\sigma$ under $K$ and $\delta$.
  \item There exists $\xi>0$ such that
  for every $K \geq \overline{K}$, there exists $\underline{\delta} \in (0,1)$ such that for every $\delta > \underline{\delta}$, there exists a PBE with strategy profile $\sigma$ such that $U_2^{\sigma} < u_2(a^*,b^*)-\xi$.
\end{enumerate}
\end{Theorem}
Theorem \ref{Theorem3} implies that longer memories of the consumers (captured by a larger $K$) may \textit{lower} consumer welfare. In particular, the consumers obtain their first-best welfare in all equilibria when $K$ is below the cutoff $\overline{K}$ but their payoffs are bounded below first best in some equilibria when $K$ is above the cutoff.

In the case where $K$ is below  $\overline{K}$, Theorems \ref{Theorem2} and \ref{Theorem3} lead to sharp predictions not only on the patient player's equilibrium payoff, but also on players' equilibrium behaviors and on the short-run players' welfare. This aspect of my result stands in contrast to most of the existing reputation results, such as those in Fudenberg and Levine (1989), that focus exclusively on the patient player's payoff but does not lead to sharp predictions on the short-run players' welfare.

Two natural questions follow from Theorems \ref{Theorem2} and \ref{Theorem3}. First, how to compute the cutoff $\overline{K}$ from the primitives $u_1$ and $u_2$? Second, how large can $\eta$ be?
My proof of Theorems \ref{Theorem2} and \ref{Theorem3} sheds light on these questions as well. To preview the answers,
I say that $a^*$ is player $1$'s \textit{optimal pure commitment action} if
\begin{equation}\label{optimalcommitment}
u_1(a^*,b^*) > \max_{a \neq a^*} \max_{b \in \textrm{BR}_2(a)} u_1(a,b).
\end{equation}
If $(u_1,u_2)$ violates (\ref{optimalcommitment}), then $\overline{K}=1$ and $\eta$ can be as large as $1$ for every $K \geq 1$. That is, the frequency with which player $1$ plays $a^*$ is $0$ in some PBEs no matter how small $K$ is.

The interesting case is the one where $(u_1,u_2)$ satisfies (\ref{optimalcommitment}), i.e., $a^*$ is player 1's optimal pure commitment action.
The cutoff $\overline{K}$ is the smallest integer $\widehat{K} \in \mathbb{N}$ such that $b^*$ best replies to the mixed action $\frac{\widehat{K}-1}{\widehat{K}} a^* + \frac{1}{\widehat{K}} a'$ for some $a' \neq a^*$.
Moreover, $\eta$ can be as large as $\frac{m}{K}$, where $m \in \{1,2,...,K\}$ is such that $b^*$ best replies to the mixed action $\frac{K-m}{K} a^* + \frac{m}{K} a'$ for some $a' \neq a^*$. As $K \rightarrow +\infty$, $\eta$ converges to $\eta^*$ where $\eta^*$ is the largest $\widetilde{\eta}$ such that $b^*$ best replies to $(1-\widetilde{\eta}) a^* + \widetilde{\eta} a'$ for some $a' \neq a^*$.
In the product choice game, players' payoffs satisfy (\ref{optimalcommitment})  since $H$ is player $1$'s highest action and committing to play $H$ results in a higher payoff for player $1$ relative to committing to play $L$. The above algorithm implies that $\overline{K} = \Big\lceil \frac{1}{1-x} \Big\rceil$.

Back to the discussions on consumer welfare, due to integer constraints, it is not necessarily the case that consumers' worst equilibrium payoff decreases in $K$. Nevertheless, as I explained earlier that when $K \rightarrow +\infty$, the lowest frequency with which player $1$ plays $a^*$
converges to $\eta^*$ where $\eta^*$ is the largest $\widetilde{\eta}$ such that $b^*$ best replies to $(1-\widetilde{\eta}) a^* + \widetilde{\eta} a'$ for some $a' \neq a^*$. Hence, there exists a uniform bound $u' < u_2(a^*,b^*)$ such that for every $K \geq \overline{K}$, there exists an equilibrium where consumer welfare is no more than $u'$.
In contrast, when $K$ is below $\overline{K}$, consumer welfare is arbitrarily close to $u_2(a^*,b^*)$ in \textit{all} equilibria.

\paragraph{Mechanism Behind Theorems 2 and 3:} I use the product choice game  to explain the mechanism behind these theorems. I argue that having a longer memory (i.e., a larger $K$) has two effects on the seller's reputational incentives.
First, when $K$ is larger, each of the seller's actions is observed by more consumers, so that he can be punished by more consumers after he shirks. This encourages him to exert high effort.

However, there is another countervailing effect, which is that a larger $K$ makes it more difficult to motivate the consumers to punish the seller.
I explain this effect using the following thought experiment.
My explanation also sheds light on the algorithm for computing $\overline{K}$ as well as how large $\eta$ can be.

Suppose the consumers believe that the strategic-type seller will play $L$ in periods $K-1,2K-1,...$ and will play $H$ in other periods. After the consumers observe that $L$ was played once in the last $K$ periods, their posterior belief assigns probability close to $\frac{1}{K}$ to $L$ being played in the current period. Hence, the consumers have an incentive to play $T$ at such histories only when $x \leq \frac{K-1}{K}$, or equivalently when $K \geq \frac{1}{1-x}$. If this is the case, then the seller prefers \textit{exerting low effort once every $K$ periods} to \textit{exerting high effort in every period}, making the consumers' belief self-fulfilling. If this is not the case (i.e., the consumers prefer to play $N$ after observing one $L$ in the last $K$ periods), then exerting low effort once every $K$ periods gives the seller a strictly lower payoff compared to exerting high effort in every period, so that in equilibrium, the consumers cannot entertain the belief that the seller will exert low effort once every $K$ periods.

Using similar ideas, one can show that for every $m \in \{1,...,K-1\}$ and $a' \neq a^*$ such that $b^*$ best replies to $\frac{K-m}{K} a^* + \frac{m}{K} a'$,  there exists an equilibrium where in every $K$ consecutive periods, player 1 plays $a'$ in $m$ periods and plays $a^*$ in $K-m$ periods, and player 2 plays $b^*$ when she observes $a'$ being played at most $m$ times and $a^*$ being played at least $K-m$ times in the last $K$ periods. This provides an explanation for how large $\eta$ can be, which I have already discussed after the statement of Theorem \ref{Theorem2}.

\paragraph{Implication on Player 1's Behavior:} Although
Theorem \ref{Theorem2} focuses on player 1's discounted action frequencies, it has implications on the dynamics of player 1's behavior and reputation.
\begin{Corollary}\label{cor2}
Suppose $1 \leq K< \overline{K}$.
For every $\varepsilon>0$, there exist a constant $C_{\varepsilon} \in \mathbb{R}_+$ and $\underline{\delta} \in (0,1)$ such that for every $\delta > \underline{\delta}$, every equilibrium under $\delta$, and every $t \in \mathbb{N}$ that satisfies $\delta^t \in (\varepsilon,1-\varepsilon)$:
\begin{itemize}
  \item[1.] The probability that $h^t \in \mathcal{H}_1^*$
  is at least $1-(1-\delta)C_{\varepsilon} $.
  \item[2.] The strategic-type player 1 plays $a^*$ with probability at least $1- (1-\delta)C_{\varepsilon}$ in period $t$.
\end{itemize}
\end{Corollary}
The proof of Corollary \ref{cor2} is in Online Appendix C, which uses Theorem \ref{Theorem2} and the no-back-loop lemma. This result implies that
in every Nash equilibrium, the strategic-type of player 1 will have a strictly positive reputation with probability close to $1$ after the initial few periods, after which he will play $a^*$ with probability close to 1 in every period until calendar time $t$ is so large that
$\delta^t$ is close to $0$. In the product choice game, it implies that the strategic-type seller will exert high effort with probability close to $1$ in all periods except for the initial few periods and periods in the distant future that have negligible payoff consequences.

\paragraph{Comparison with Fudenberg and Levine (1989):} When $K < \overline{K}$, Theorem \ref{Theorem2} leads to a sharp prediction on player $1$'s behavior, that player 1 will play $a^*$ with frequency arbitrarily close to $1$ in \textit{all} equilibria.
This conclusion stands in contrast to the reputation model of Fudenberg and Levine (1989) in which the short-run players can observe the entire history of player 1's actions and there is \textit{a lack-of sharp prediction} in terms of player 1's behaviors. For example, in the product choice game, Li and Pei (2021) show that the discounted frequency with which the seller exerts high effort can be anything between $\frac{x}{1+(1-x)c_T}$ and $1$ in Fudenberg and Levine (1989)'s reputation model, so a wide range of behaviors can occur in equilibrium.\footnote{In Online Appendix D, I compare the predictions on player $1$'s \textit{discounted action frequency} in Fudenberg and Levine (1989)'s model to those in my model under any arbitrary $K \in \mathbb{N}$. I show that my model leads to sharper predictions in terms of the action frequencies as long as $a^*$ is player $1$'s optimal pure commitment action.}

According to
Theorem \ref{Theorem2}, when the consumers have short memories and only receive \textit{coarse information} about the seller's past actions,
the only equilibria that survive are those where the seller exerts high effort in almost all periods.
It implies that the bad equilibria (i.e., those where the consumers receive a low payoff) in Fudenberg and Levine (1989)'s model rely on consumer-strategies that depend either on events that happened in the distant past or on the fine details of the game's history. These strategies sound less plausible (i) in markets without good record-keeping institutions where it is hard for the consumers to learn the exact sequence of the seller's actions, (iii) in online platforms that only disclose some aggregate statistics of the seller's recent performances but do not disclose the exact timing of each individual rating, and
(iii) when consumers have limited capacity to process detailed information about the exact sequence of the seller's actions in which case 
their decisions are based on coarse summary statistics.

\paragraph{Proof Sketch:} The idea behind the proof of the second part of these theorems is contained in the thought experiment, with details in Appendix \ref{subB.2}.
My technical contribution is in the proof of the first part.

A major challenge is that characterizing
\textit{all} equilibria in an infinitely repeated game is not tractable in general, and
ruling out the type of equilibria constructed in the proof for the second part is \textit{insufficient} to show that player 1 will play $a^*$ with frequency close to $1$ in \textit{all} equilibria. One of the reasons is that a Bayesian short-run player's expectation of the patient player's current-period action \textit{may not} be close to the empirical frequency of the patient player's last $K$ actions.
For example, in the product choice game,  the seller may play $H$ with high probability at histories where he played $L$ in all of the last $K$ periods.

 I sketch the proof and the omitted details are in Appendix \ref{subB.1}. Some of my arguments such as Lemma \ref{L4.1}, \ref{Linflow}, and \ref{L4.3} apply independently of players' incentives and payoff functions, and therefore, are portable to other repeated game settings with limited memories and more generally, repeated games where players use finite automaton strategies.
Readers who are not interested in the proof can jump to Section \ref{sec4}.

I define a \textit{state} as a sequence of player 1's actions with length $K$, i.e., $(a_{t-K},...,a_{t-1})$.
Let $S \equiv A^K$ be the \textit{set of states} with $s \in S$. Let $s^* \equiv (a^*,...,a^*)$.
Fix a strategy profile $\sigma \equiv (\sigma_1,\sigma_2)$. For every $s \in S$, let $\mu(s)$ be the probability that the current-period state is $s$ conditional on the event that player 1 is the strategic type and calendar time is at least $K$. For every pair of states $s,s' \in S$, let $Q(s \rightarrow s')$ be the probability that the state in the next period is $s'$ conditional on the state in the current period is $s$, player 1 is the strategic type, and the calendar time is at least $K$. Let $p(s)$ be the probability that the state is $s$ conditional on calendar time being $K$ and player 1 is the strategic type.
The goal is to show that $\mu(s^*)$ is close to $1$ in all equilibria.
I state three lemmas which hold for \textit{all} strategy profiles and \textit{all} stage-game payoffs, i.e., they are portable to other repeated games with limited memories or when players are required to use finite automaton strategies.
\begin{Lemma}\label{L4.1}
For any $\delta \in (0,1)$ and any equilibrium under $\delta$, we have
\begin{equation}\label{4.8}
  \mu(s')= (1-\delta) p(s') + \delta \sum_{s \in S} \mu(s) Q(s \rightarrow s') \textrm{ for every } s' \in S.
\end{equation}
\end{Lemma}
The proof is in Appendix \ref{subB.1}.
Intuitively, Lemma \ref{L4.1} implies that the occupation measure of every state $s'$ is a convex combination of its probability in period $K$ and the expected probability that the state moves to $s'$ after  period $K$. For any non-empty subset of states $S' \subset S$, let
\begin{equation}\label{in}
\mathcal{I}(S') \equiv  \sum_{s' \in S'} \sum_{s \notin S'} \mu(s) Q(s \rightarrow s')
\end{equation}
be the \textit{inflow} to $S'$ from states that do not belong to $S'$, and let
\begin{equation}\label{out}
 \mathcal{O} (S')  \equiv \sum_{s' \in S'} \sum_{s \notin S'}  \mu(s') Q(s' \rightarrow s)
\end{equation}
be the \textit{outflow} from $S'$ to states that do not belong to $S'$.
By definition,
\begin{equation*}
\sum_{s' \in S'}
    \mu(s')=
    \sum_{s' \in S'}
    \mu(s')\underbrace{\Big( \sum_{s \in S'} Q(s' \rightarrow s) +  \sum_{s \notin S'} Q(s' \rightarrow s)\Big)}_{=1}
    =
    \sum_{s' \in S'} \sum_{s \in S'} \mu(s') Q(s' \rightarrow s)
    +\underbrace{\sum_{s' \in S'} \sum_{s \notin S'}  \mu(s') Q(s' \rightarrow s)}_{\equiv \mathcal{O}(S')}.
\end{equation*}
For any $S' \subset S$, by summing up the two sides of equation (\ref{4.8}) for all $s' \in S'$, one can obtain that
\begin{equation*}
\sum_{s' \in S'} \sum_{s \in S'} \mu(s') Q(s' \rightarrow s) +\underbrace{ \sum_{s' \in S'} \sum_{s \notin S'} \mu(s) Q(s \rightarrow s')}_{\equiv \mathcal{I}(S')}=
\sum_{s' \in S'}
    \mu(s')+
\sum_{s' \in S'} \frac{1-\delta}{\delta} \Big\{
    \mu(s')-p(s')
    \Big\}.
\end{equation*}
These equations imply that $\mathcal{I}(S')=\mathcal{O}(S') + \sum_{s' \in S'} \frac{1-\delta}{\delta} \Big\{
    \mu(s')-p(s')
    \Big\}$. Since $\mu$ and $p$ are probability measures on $S$, we have $|\sum_{s' \in S'} (\mu(s')-p(s')) | \leq 1$. This leads to
    the following lemma:
\begin{Lemma}\label{Linflow}
For every non-empty subset $S' \subset S$, we have:
\begin{equation}\label{4.9}
|\mathcal{I}(S')-\mathcal{O}(S')| = \Big| \sum_{s' \in S'} \frac{1-\delta}{\delta} \big(
    \mu(s')-p(s')
    \big)  \Big| \leq \frac{1-\delta}{\delta}.
\end{equation}
\end{Lemma}
I partition the set of states according to $S \equiv S_0 \cup ... \cup S_K$ so that $S_k$ is the set of states where $k$ of the player $1$'s last $K$ actions were not
$a^*$. By definition, $S_0=\{s^*\}$. I further partition every $S_k \equiv \cup_{j=1}^{J(k)} S_{j,k}$ according to player $2$'s information structure, i.e., two states belong to the same partition element $S_{j,k}$ if and only if player 2 distinguish between these two states. For every state that belongs to $S_k$, exactly one of the following two statements is true, depending on whether player 1's action $K$ periods ago was $a^*$:
\begin{enumerate}
  \item The state in the next period belongs to $S_{k-1}$ or $S_k$, depending on player 1's current-period action.
  \item The state in the next period belongs to $S_{k}$ or $S_{k+1}$, depending on player 1's current-period action.
\end{enumerate}
Therefore, I partition each $S_{j,k}$ into  $S_{j,k}^*$ and $S_{j,k}'$ such that for every $s \in S_{j,k}$, $s \in S_{j,k}^*$ if and only if player 1's action $K$ periods ago was $a^*$, and $s \in S_{j,k}'$ otherwise. For any $S',S'' \subset S$ with $S' \cap S'' =\emptyset$, let
\begin{equation}\label{flow}
    \mathcal{Q}(S' \rightarrow S'') \equiv \sum_{s' \in S'} \sum_{s'' \in S''} \mu(s') Q(s' \rightarrow s'')
\end{equation}
be the expected flow from $S'$ to $S''$. According to Bayes rule, upon observing a state that belongs to $S_{j,k}$, player 2 believes that
player 1's action is $a^*$ with probability
\begin{equation}\label{4.14}
  \mathcal{Q} (S_{j,k} \rightarrow S_{k-1})+\sum_{s \in S_{j,k}^*} \sum_{s' \in S_{j,k}} \mu(s) Q(s \rightarrow s'),
\end{equation}
and is not $a^*$ with probability
  $\mathcal{Q}(S_{j,k} \rightarrow S_{k+1})+ \sum_{s \in S_{j,k}'} \sum_{s' \in S_k} \mu(s) Q(s \rightarrow s')$.
The next lemma shows that as long as both $\mathcal{Q}(S_{k-1} \rightarrow S_{j,k})$ and $\mathcal{Q}(S_{j,k} \rightarrow S_{k-1})$ are bounded above by a linear function of $1-\delta$, either $\sum_{s \in S_{j,k}} \mu(s)$ is also bounded above by some linear function of $1-\delta$, or player 2 believes that player $1$ will play $a^*$ with probability strictly less than
$\frac{K-1}{K}$ upon observing that the current state belongs to
$S_{k,j}$.
\begin{Lemma}\label{L4.3}
For every $z \in \mathbb{R}_+$, there exist $y \in \mathbb{R}_+$ and $\underline{\delta} \in (0,1)$ such that for every equilibrium under $\delta > \underline{\delta}$ and every $S_{j,k}$ with $k \geq 1$. If $\max \{ \mathcal{Q} (S_{j,k} \rightarrow S_{k-1}),\mathcal{Q} (S_{k-1} \rightarrow S_{j,k}) \} \leq z(1-\delta)$, then either
\begin{equation}\label{bound}
    \sum_{s \in S_{j,k}} \mu(s) \leq y (1-\delta)
\end{equation}
or
\begin{equation}\label{incentive}
    \frac{\displaystyle \mathcal{Q} (S_{j,k} \rightarrow S_{k-1})+\sum_{s \in S_{j,k}^*} \sum_{s' \in S_{j,k}} \mu(s) Q(s \rightarrow s')}
    {\displaystyle \mathcal{Q}(S_{j,k} \rightarrow S_{k+1})+ \sum_{s \in S_{j,k}'} \sum_{s' \in S_k} \mu(s) Q(s \rightarrow s') }
    <K-1.
\end{equation}
\end{Lemma}
The proof requires some heavy algebra, which is in Appendix \ref{subB.1}. Intuitively, in the case where (\ref{bound}) is satisfied, the states in $S_{j,k}$ have negligible occupation measure as $\delta \rightarrow 1$. In the case where (\ref{incentive}) is satisfied and $K < \overline{K}$, player 2 has no incentive to play actions $b^*$ or above after observing any state in $S_{j,k}$.

The next two steps make use of players' incentive constraints. First, I use Lemma \ref{Linflow} and
the no-back-loop lemma to show that
both the inflow to $S_0 \equiv \{s^*\}$ and the outflow from $S_0$ are small.
\begin{Lemma}\label{L4.2}
For every $\delta$ and in every Nash equilibrium under $\delta$, $\mathcal{O}(S_0) \leq \frac{2(1-\delta)}{\delta}$ and  $\mathcal{I}(S_0) \leq \frac{1-\delta}{\delta}$.
\end{Lemma}
\begin{proof}
Let $S'$ denote a subset of states such that $s \in S'$ if and only if (i) $s \neq s^*$, and (ii) there exists a best reply $\widehat{\sigma}_1$ such that $s^*$ is reached within a finite number of periods when the initial state is $s$.
The no-back-loop lemma implies that at least one of the two statements is true:
\begin{enumerate}
  \item Player 1 has no incentive to play actions other than $a^*$ at $s^*$.
  \item Player $1$ has an incentive to play actions other than $a^*$ at $s^*$, and as long as player $1$ plays any such best reply, the state never reaches $S'$ when the initial state is $s^*$.
\end{enumerate}
In the first case, $\mathcal{O}(\{s^*\})=0$, and Lemma \ref{Linflow} implies that $\mathcal{I}(\{s^*\}) \leq \frac{1-\delta}{\delta}$. In the second case, the definition of $S'$ implies that $\mathcal{I}(S')=0$. According to Lemma \ref{Linflow},  $\mathcal{O}(S') \leq \frac{1-\delta}{\delta}$. The definition of $S'$  implies that $\mathcal{I}(\{s^*\}) \leq \mathcal{O}(S')$, so $\mathcal{I}(\{s^*\}) \leq \frac{1-\delta}{\delta}$. According to Lemma \ref{Linflow},
$\mathcal{O}(\{s^*\}) \leq \frac{2(1-\delta)}{\delta}$.
\end{proof}
Lemma \ref{L4.2} implies that the inflow to $S_0$ and the outflow from $S_0$ are both negligible. Lemma \ref{L4.3} implies that for every $S_{j,1}$, either the occupation measure of states in $S_{j,1}$ is negligible, or player $2$ has no incentive to play $b^*$ at $S_{j,1}$. The next lemma shows that
it \textit{cannot} be the case that states in $S_{1}$ have significant occupation measure yet player $2$ has no incentive to play $b^*$ at $S_{j,1}$.
This conclusion generalizes to every $S_{j,k}$, provided that the flow from $S_{k-1}$ to $S_{j,k}$ and that from $S_{j,k}$ to $S_{k-1}$ are both negligible. Iteratively apply Lemma \ref{L4.3} and Lemma \ref{L4.4}, all states except for $s^*$ (the unique state in $S_0$) have negligible occupation measure.
\begin{Lemma}\label{L4.4}
Suppose $u_2$ is such that $b^*$ does not best reply to $\frac{K-1}{K}a^*+ \frac{1}{K} a'$ for every $a' \neq a^*$.
For every $y >0$, there exist $z>0$ and $\underline{\delta} \in (0,1)$ such that for every $\delta > \underline{\delta}$, every Nash equilibrium under $\delta$, and every $k \in \{1,2,...,K\}$. If $\max \{ \mathcal{Q}(S_{k-1} \rightarrow S_{k}), \mathcal{Q}(S_{k} \rightarrow S_{k-1}) \} < y(1-\delta)$, then
$\sum_{s \in S_{k}} \mu(s) \leq  z(1-\delta)$.
\end{Lemma}
The proof is relegated to Appendix \ref{subB.1}. The intuition is that when both $\mathcal{Q}(S_{k-1} \rightarrow S_{k})$ and $\mathcal{Q}(S_{k} \rightarrow S_{k-1})$ are negligible, $\mathcal{Q}(S_{k-1} \rightarrow S_{j,k})$ and $\mathcal{Q}(S_{j,k} \rightarrow S_{k-1})$ are also negligible for every $j$. Suppose by way of contradiction that
$\sum_{s \in S_{k}} \mu(s)$ is bounded away from $0$, then Lemma \ref{L4.3} implies that at every $S_{j,k}$ where
states in $S_{j,k}$ occur with  occupation measure bounded above $0$, player $2$ has no incentive to play $b^*$ at $S_{j,k}$. Since the flow from $S_{k}$ to $S_{k-1}$ is negligible, the flow from $S_k$ to $S_{k+1}$ must be bounded above $0$.
Lemma \ref{Linflow} then implies that the flow from $S_{k+1}$ to $S_k$ is also bounded above $0$. This implies that in equilibrium, player $1$ will take the most costly action $a^*$ at some states in $S_{k+1}$ in order to reach states in $S_{j,k}$ where player $2$ has no incentive to play $b^*$, i.e., player $1$'s payoff is bounded below $u_1(a^*,b^*)$ in states that belong to $S_{j,k}$.
This is suboptimal for player $1$ since he can secure payoff $u_1(a^*,b^*)$ by playing $a^*$ in every period, which implies that
for every $k \geq 1$, $\sum_{s \in S_{k}} \mu(s) \approx 0$ in all equilibria.

\section{Extensions}\label{sec4}
Section \ref{sub4.1} examines other prior beliefs about calendar time, for example, when the game's continuation probability is different from player 1's discount factor. Section \ref{sub4.2} extends my results to games where $|B| \geq 3$.
Section \ref{sub4.4} extends my results to situations where each player $2$ can only observe the summary statistics of some noisy signals about player 1's actions.
Section \ref{sub4.3} extends my results to situations where player 2 observes coarse summary statistics about player 1's last $K$ actions and provides conditions under which coarsening the summary statistics observed by the consumers can improve their welfare.

\subsection{Player 2's Prior Belief about Calendar Time}\label{sub4.1}
The parameter $\delta$ plays two roles in my baseline model: It is both the patient player's discount factor as well as
the probability with which the game continues after each period. The latter affects the short-run players' incentives
through their prior beliefs about calendar time.

I extend my results to environments where the patient player's discount factor do not coincide with the game's continuation probability, and more generally, to environments where the short-run players have alternative prior beliefs about calendar time.
For example, suppose the short-run players' prior belief assigns probability $(1-\overline{\delta})\overline{\delta}^t$ to calendar time being $t \in \mathbb{N}$, and the patient player's discount factor is $\delta$, which I assume is no more than $\overline{\delta}$, i.e., $1>\overline{\delta}>\delta>0$.
It models situations where
the patient player discounts future payoffs for two reasons. First, he is indifferent between receiving one unit of utility in period $t$ and receiving $\delta/\overline{\delta}$ unit of utility in period $t-1$. Second,
the game ends with probability $1-\overline{\delta}$ after each period.

My \textit{no-back-loop lemma} extends to this setting. This is because that conclusion applies independently of player 1's discount rate and the game's continuation probability. The statement of Theorem \ref{Theorem1} is modified as follows: Suppose $(u_1,u_2)$ satisfies Assumption \ref{Ass1} and the continuation probability $\overline{\delta}$ is large enough such that all of player 2's best replies to the mixed action
     $\big\{1-\frac{(1-\overline{\delta})(1-\pi_0)}{\pi_0}
     \big\} a^* +  \frac{(1-\overline{\delta})(1-\pi_0)}{\pi_0} \underline{a}$
are no less than $b^*$, then player 1's payoff in every equilibrium is at least
   $ (1-\delta^K) u_1(a^*,\underline{b})+\delta^K u_1(a^*,b^*)$.

Hence, Theorem \ref{Theorem1} applies to \textit{any} discount factor of the patient player, as long as the probability with which the game continues after each period is above some cutoff. Moreover, that cutoff depends \textit{only} on the prior probability of commitment type $\pi_0$ and player 2's stage-game payoff function $u_2$. Intuitively, since the patient player's equilibrium strategy satisfies the no-back-loop property, there is at most one period over the infinite horizon in which he has a positive reputation yet he plays an action other than $a^*$.
Therefore, every short-run player has a strict incentive to play $b^*$ after observing $a^*$ in the last $K$ periods when her prior belief assigns a low enough probability to each calendar time. The latter is the case when $\overline{\delta}$ is large.

In order to state Theorem \ref{Theorem2} in this general setting, I redefine $F^{\sigma} (a,b)$ based on
the game's continuation probability $\overline{\delta}$:
\begin{equation*}
   F^{\sigma}(a,b) \equiv \mathbb{E}^{\sigma}\Big[
   \sum_{t=0}^{+\infty} (1-\overline{\delta})\overline{\delta}^t \mathbf{1}\{a_t=a,b_t=b\}
  \Big].
\end{equation*}
The motivation for defining $F^{\sigma}(a,b)$ in this way is that under the interpretation that the game ends after the patient player exits, the (normalized) expected sum of the short-run players' payoff is:
\begin{equation}\label{welfare}
U_2^{\sigma} \equiv \mathbb{E}^{\sigma}\Big[
   \sum_{t=0}^{+\infty} (1-\overline{\delta})\overline{\delta}^t u_2(a_t,b_t)
\Big] = \sum_{(a,b) \in A \times B} F^{\sigma}(a,b) u_2(a,b).
\end{equation}
Hence, $U_2^{\sigma}$ depends on the strategy profile $\sigma$ only through $\{F^{\sigma}(a,b)\}_{(a.b) \in A \times B}$.
I restate Theorem \ref{Theorem2}. Suppose $(u_1,u_2)$ satisfies Assumption \ref{Ass1} and the discount factor $\delta$ is larger than some cutoff $\underline{\delta} \in (0,1)$,
\begin{enumerate}
\item Suppose $a^*$ is player 1's optimal pure commitment action and $b^*$ does not best reply to the mixed action $\frac{K-1}{K}a^*+ \frac{1}{K} a'$ for every $a' \neq a^*$. Then there exists a constant $C \in \mathbb{R}_+$ that is independent of $\delta$ and $\overline{\delta}$ such that
$F^{\sigma}(a^*,b^*) \geq 1-(1-\overline{\delta})C$
for every Nash equilibrium $\sigma$ under discount factor $\delta$.
\item Suppose either $a^*$ is not player 1's optimal pure commitment action, or $b^*$ best replies to $\frac{K-1}{K}a^*+ \frac{1}{K} a'$ for some $a' \neq a^*$.
There exist $\underline{\delta} \in (0,1)$ and $\eta>0$ such that for every $\delta > \underline{\delta}$, there exists an equilibrium $\sigma$ such that $\sum_{b \in B} F^{\sigma}(a^*,b) \leq 1-\eta$.
\end{enumerate}

Hence, the occupation measure of $(a^*,b^*)$ is arbitrarily close to $1$ in all equilibria even when player 1's discount factor $\delta$ is bounded away from $1$. The important parameter is $\overline{\delta}$, the game's continuation probability, which affects player 2's belief about calendar time.
When $K$ is small, the total occupation measure of action profiles other than  $(a^*,b^*)$ is bounded above by some linear function of $1-\overline{\delta}$, i.e., it vanishes to zero as long as the game continues after each period with probability arbitrarily close to $1$.

\paragraph{Remark:} My no-back-loop lemma is about the patient player's best reply, which holds for \textit{all} prior beliefs about calendar time. The proofs of my theorems only use two properties of player 2's prior belief about calendar time. The proof of Theorem \ref{Theorem1} uses the property that the prior probability of each individual calendar time being close to $0$, in which case the no-back-loop lemma implies that player 2 believes that $a^*$ will occur with probability close to $1$ in the current period after observing $a^*$ being played in all of the last $K$ periods. The proof of
Theorems \ref{Theorem2} and \ref{Theorem3} uses the property that the ratio between the prior probability of $t$ and that of $t+1$ being close to $1$ for every $t \in \mathbb{N}$. If this is the case, then in an equilibrium where the rational type player $1$ plays $a' (\neq a^*)$ once every $K$ periods, player 2's posterior belief assigns probability close to $1/K$ to player $1$'s current period action being $a'$ after observing $a'$ occurred only once in the last $K$ periods.

\subsection{Games where $|B| \geq 3$}\label{sub4.2}
I generalize my theorems to games where player $2$ has three or more actions, i.e., $|B| \geq 3$. In contrast to games where $|B|=2$, one cannot rank any pair of player $2$'s mixed actions via FOSD. For example, when $B =\{b_1,b_2,b_3\}$ with $b_1 \succ_B b_2 \succ_B b_3$, mixed actions $\frac{1}{2}b_1+ \frac{1}{2}b_3$ and $b_2$ cannot be ranked via FOSD. Nevertheless, this issue does not arise in games where player $2$'s \textit{mixed-strategy best replies} can be ranked according to FOSD, which is satisfied by most of the games studied in the reputation literature.
Let
\begin{equation}\label{bestreplyset}
\mathcal{B}^* \equiv \Big\{\beta \in \Delta (B) \Big| \textrm{ there exists } \alpha \in \Delta (A) \textrm{ such that } \beta \textrm{ best replies to } \alpha \Big\}.
\end{equation}
\begin{Assumption}\label{Ass2}
For every $\beta,\beta' \in \mathcal{B}^*$, either $\beta \succeq_{FOSD} \beta'$ or $\beta' \succeq_{FOSD} \beta$ or both.
\end{Assumption}
Assumption \ref{Ass2} has no bite when $|B|=2$. In games where $|B| \geq 3$,
Assumption \ref{Ass2} is satisfied in the class of games studied by Liu and Skrzypacz (2014): They assume that player $2$ has a unique best reply to every $\alpha \in \Delta (A)$, which is the case when $u_2(a,b)$ is strictly concave in $b$. In the games studied by their paper, all  actions in $\mathcal{B}^*$ are pure, in which case any pair of them can be ranked according to FOSD.

A more general sufficient condition for Assumption \ref{Ass2} is that player $2$ has a \textit{single-peaked preference} over her actions regardless of her belief about player $1$'s action, i.e., $u_2(a,b)$ is strictly quasi-concave in $b$. This is because when player $2$'s preference is single-peaked, either she has a unique best reply to $\alpha$, or she has two pure-strategy best replies to $\alpha$ which are adjacent elements in set $B$. Quah and Strulovici (2012) provide a full characterization of this sufficient condition using
the well-known single-crossing property.

I restate
the no-back-loop lemma for games that satisfy Assumptions \ref{Ass1} and \ref{Ass2}: For every $\sigma_2: \mathcal{H}_2 \rightarrow \mathcal{B}^*$ and pure strategy $\widehat{\sigma}_1: \mathcal{H}_1 \rightarrow A$ that best replies to $\sigma_2$, there is \textit{no} $h^t \in \mathcal{H}_1(\widehat{\sigma}_1,\sigma_2) \bigcap \mathcal{H}_1^*$
such that strategy $\widehat{\sigma}_1$ plays an action that is not $a^*$ at $h^t$ and reaches a history in $\mathcal{H}_1(\widehat{\sigma}_1,\sigma_2) \bigcap \mathcal{H}_1^*$ in the future.
One can further strengthen the no-back-loop result to every $\sigma_2 : \mathcal{H}_2 \rightarrow \Delta (B)$ and pure strategy $\widehat{\sigma}_1: \mathcal{H}_1 \rightarrow A$ such that $\widehat{\sigma}_1$ best replies to $\sigma_2$ and $\sigma_2(h_2^t) \in \mathcal{B}^*$ for every $h_2^t$ that occurs with positive probability under $(\widehat{\sigma}_1,\sigma_2)$.

My proof of the no-back-loop result in Appendix \ref{secA} covers these cases.
Intuitively, since we only consider player 2's strategies where she plays a potentially mixed action that belongs to $\mathcal{B}^*$ at every history, player 2's actions at any pair of histories can be ranked according to FOSD.
Assumption \ref{Ass2} is used when arguing that player $2$'s action at the white circle weakly FOSDs her action at the green circle, since the Topkis Theorem only implies that her action at the green circle \textit{cannot} strictly FOSD her action at the white circle.
After establishing this generalized no-back-loop lemma, Theorem \ref{Theorem1} can be shown using the same argument as that in Section \ref{sub3.2}. This is because at every history that occurs with positive probability,
player $2$'s mixed action at that history must be a best reply to some $\alpha \in \Delta (A)$, which means that her action belongs to
 $\mathcal{B}^*$.
The proof of Theorems \ref{Theorem2} and \ref{Theorem3} in Appendix \ref{secB} covers games where $|B| \geq 3$ and $u_2$
satisfies Assumption \ref{Ass2}.

\subsection{Noisy Information about Player 1's Actions}\label{sub4.4}
My baseline model rules out imperfect monitoring by assuming that each consumer observes the number of times that the seller took each of his actions in the last $K$ periods. In practice, consumers learn from previous consumers' experiences or from online ratings, which might be \textit{noisy signals} of the seller's actions.

I extend my theorems to environments where the consumers receive \textit{noisy signals} about the seller's past actions.  
Formally, let $\widetilde{a}_t \in A$ be a signal of player $1$'s action $a_t$ such that $\widetilde{a}_t=a_t$ with probability $1-\varepsilon$, and $\widetilde{a}_t$ is drawn from some distribution $\alpha \in \Delta (A)$ with probability $\varepsilon$. 
Player $1$ observes the entire history $h^t=\{a_s,b_s,\widetilde{a}_s\}_{s=0}^{t-1}$. Player $2_t$  only observes the number of times that each \textit{signal realization} occurred in the last $\min \{t,K\}$ periods. In the product choice game, when the seller exerts high effort, the consumer will have a good experience with probability $1-\varepsilon \alpha(H)$, and when the seller exerts low effort, the consumer will have a bad experience with probability $1-\varepsilon \alpha(L)$.
My baseline model focuses on the special case where
$\varepsilon=0$, i.e., $\widetilde{a}_t=a_t$ with probability $1$.

I show that a version of my no-back-loop lemma holds for all  small enough $\varepsilon$.
For every $t \geq K$, player $1$'s incentive in period $t$ depends on the history only through $(\widetilde{a}_{t-K},...,\widetilde{a}_{t-1})$. Let $S \equiv A^K$ be the set of signal vectors of length $K$ with a typical element denoted by $s \in S$, which I call a \textit{state}. Without loss of generality, I focus on player $1$'s strategies that are measurable with respect to the state, i.e., $\sigma_1: S \rightarrow \Delta (A)$. I say that a pure strategy $\widehat{\sigma}_1: S \rightarrow A$ induces an \textit{$\varepsilon$-back-loop} if there exist a subset of states $\{s_0,...,s_{M}\}  \subset S$ such that $s_0=s_M=(a^*,...,a^*)$ and for every $i \in \{0,1,...,M-1\}$, if player $1$ plays $\widehat{\sigma}_1 (s_i)$ in state $s_i$, then it reaches state $s_{i+1}$ in the next period with probability more than $1-\varepsilon$.
\begin{Corollary}\label{cor5}
There exists $\varepsilon>0$ such that for every $\alpha \in \Delta (A)$ and $\sigma_2: \mathcal{H}_2 \rightarrow \Delta (B)$, if a pure strategy $\widehat{\sigma}_1$ best replies to $\sigma_2$, then $\widehat{\sigma}_1$ does not induce any $\varepsilon$-back-loop. If player $2$ has three or more actions and $u_2$ satisfies Assumption \ref{Ass2}, then this conclusion holds for all $\sigma_2: \mathcal{H}_2 \rightarrow \mathcal{B}^*$.
\end{Corollary}
The proof is in Online Appendix E, which is similar to that of the no-back-loop lemma in the baseline model. This corollary implies that although back loops may occur with positive probability due to noisy signals, player $1$ \textit{has no intention} to induce any back loop as long as he plays a best reply. This implies that when $\varepsilon$ is close to $0$, the probability of back loops is also close to $0$.

The same argument can be used to show Theorem \ref{Theorem1}, that player $1$ can guarantee himself a payoff of (approximately) at least $u_1(a^*,b^*)$ when $\delta$ is close enough to $1$ and $\varepsilon$ is small enough. This is because conditional on observing
$(\widetilde{a}_{t-K},...,\widetilde{a}_{t-1})=(a^*,...,a^*)$,
 the probability that player $2$ assigns to player $1$ playing $a^*$ in the current period is close to $1$. This implies that player $2$'s action is at least $b^*$, and therefore, player $1$'s payoff is approximately $u_1(a^*,b^*)$ when he plays $a^*$ in every period. Theorem \ref{Theorem2} can also be extended to environments where $\varepsilon$ is small, i.e., the same cutoff $\overline{K}$ applies as long as $\varepsilon$ is small enough. Intuitively, this is because when $\varepsilon$ is small, the occupation measure of states other than $(\widetilde{a}_{t-K},...,\widetilde{a}_{t-1})=(a^*,...,a^*)$ is close to $0$ since one can show that (i) when $K$ is below the cutoff, the probability that other states being generated by player $1$'s deliberate behavior is close to $0$, and (ii) when $\varepsilon$ is close to $0$, the probability that other states being generated by noise in player 2's signal is also close to $0$.

\subsection{Learning from Coarse Summary Statistics}\label{sub4.3}
This section studies an extension where the consumers can only learn from \textit{coarse summary statistics} about the seller's last $K$ actions.
Following Acemoglu, Makhdoum, Malekian and Ozdaglar (2022), a \textit{coarse summary statistics} is characterized by a partition of $A \equiv A_1\cup...\cup A_n$, so that for every $t \in \mathbb{N}$, the short-run player who arrives in period $t$ only observes the number of times player 1's actions belong to each partition element $A_i$ in the last $\min \{t,K\}$ periods.
My baseline model considers the finest partition of $A$, i.e., player $2$ observes the number of times that player $1$ took each of his actions in the last $K$ periods. Under the coarsest partition of $A$, player 2 receives no information about player 1's past actions.

This extension fits when the consumers do not communicate precise information about the seller's action to future consumers. Instead, they can only describe which of the several broad categories the seller's action belongs to. For example, each consumer only tells future consumers whether she had a good experience or a bad experience, but she finds it too time-consuming to precisely describe the seller's action (e.g., exactly how good or how bad), particularly when the cardinality of $A$ is large.

Since there exists a complete order $\succ_A$ on $A$ and $u_1(a,b)$ is strictly decreasing in $a$, for every partition element $A_i$, the strategic-type player $1$ will never choose action $a \in A_i$ if there exists $a' \in A_i$ that satisfies $a \succ_A a'$. Hence, analyzing the game under an $n$-partition $\{A_1,...,A_n\}$ of $A$ is equivalent to analyzing a game where player 1's action set only contains the following $n$ actions: $\{\min A_1,...,\min A_n\}$.

When players' stage-game payoffs satisfy Assumption \ref{Ass1}, player 2 has no incentive to play $b^*$ unless player 1 plays $a^*$ with positive probability. When the prior probability of commitment type $\pi_0$ is small enough such that player 2 has no incentive to play $b^*$ when player 1 plays $a^*$ with probability no more than $\pi_0$, the strategic-type player 1 has no incentive to play $a^*$ and player $2$ has no incentive to play $b^*$ unless the partition element that contains $a^*$ is a singleton. If we partition $A$ according to $A =\{a^*\} \bigcup \Big( A\backslash\{a^*\}\Big)$, i.e., consumers only observe the number of times the seller chose $a^*$ in the last $K$ periods but cannot distinguish other actions, then consumers may receive a higher welfare under some intermediate $K$. Intuitively,  such a partition helps the seller to credibly commit not to take any action other than his commitment action $a^*$ and his lowest-cost action $\underline{a} \equiv \min A$. This provides consumers a stronger incentive to punish the seller after the seller loses his reputation, since consumers know that the seller will take the lowest action as long as he does not take the highest action. Such an effect
motivates the seller to play $a^*$ in every period.
\begin{Corollary}\label{Cor4}
Suppose players' stage-game payoffs $(u_1,u_2)$ satisfy Assumptions \ref{Ass1} and \ref{Ass2},
\begin{enumerate}
  \item If the partition element that contains $a^*$ is not a singleton, then the discounted frequency with which the strategic-type of player $1$ plays $a^*$ is $0$ in all Nash equilibria.
  \item If the partition element that contains $a^*$ is a singleton (without loss of generality, let $A_1 \equiv \{a^*\}$), then when $\delta > \underline{\delta}(\pi_0)$, the strategic-type player 1's payoff is at least $(1-\delta^K) u_1(a^*,\underline{b})+\delta^K u_1(a^*,b^*)$ in every Nash equilibrium. Furthermore, there exists an integer $\overline{K} \in \mathbb{N}$ such that
  \begin{itemize}
    \item[(i)] There exists $C \in \mathbb{R}_+$ that is independent of $\delta$ such that for every  $1 \leq K < \overline{K}$, we have
    $F^{\sigma}(a^*,b^*) \geq 1-(1-\delta)C$ for every Nash equilibrium $\sigma$ under $K$ and $\delta$.
    \item[(ii)] There exists $\eta>0$ such that for every $K \geq \overline{K}$, there exists $\underline{\delta} \in (0,1)$ such that for every $\delta > \underline{\delta}$, there exists a PBE such that $\sum_{b \in B} F^{\sigma}(a^*,b) \leq 1-\eta$.
  \end{itemize}
\end{enumerate}
\end{Corollary}
Corollary \ref{Cor4} directly follows from Theorems \ref{Theorem1} and \ref{Theorem2}, the proof of which is omitted in order to avoid repetition. The way to compute the cutoff $\overline{K}$ is similar to that in the baseline model. If (\ref{optimalcommitment}) is violated, then $\overline{K}=1$ and there exists an equilibrium in which the strategic type plays $a^*$ with zero frequency. If (\ref{optimalcommitment}) is satisfied, then $\overline{K}$ is the smallest $K$ such that $b^*$ best replies to the mixed action
$\frac{K-1}{K}a^*+  \frac{1}{K} \min_{j \in \{2,...,n\}} \{ \min A_j \}$.

Hence, for any $K \in \mathbb{N}$,
if there exists a partition of $A$ under which player 1 plays $a^*$ with frequency arbitrarily close to one in all equilibria, then
player 1 plays $a^*$ with frequency arbitrarily close to one in all equilibria under partition $A =\{a^*\} \bigcup \Big( A\backslash\{a^*\}\Big)$. That is to say, if the objective is to maximize the consumers' payoffs in the \textit{worst} equilibrium, then it is optimal to disclose to consumers only the number of times that the seller chose the commitment action $a^*$ in the last $K$ periods.

Corollary \ref{Cor4} also implies that coarsening the summary statistics  \textit{cannot} improve consumers' welfare when $|A|=2$. However, doing so may improve consumers' welfare when $|A| \geq 3$ and $K$ is intermediate such that (i)  $b^*$ does not best reply to $\frac{K-1}{K}a^*+ \frac{1}{K} \underline{a}$, and (ii) $b^*$ is a strict best reply to $\frac{K-1}{K}a^*+ \frac{1}{K} a'$ for some $a' \notin \{a^*,\underline{a}\}$. The intuition is that by pooling actions other than $a^*$, consumers believe that the seller's action is his lowest action $\underline{a}$ whenever his action is not $a^*$. This provides consumers stronger incentives to punish the seller after the seller loses his reputation. This in turn encourages the seller to sustain his reputation.

\section{Discussions}\label{sec5}
My baseline model makes the standard assumptions that (i) $u_1(a,b)$ is strictly increasing in $b$ and is strictly decreasing in $a$, and (ii) $u_2(a,b)$ has strictly increasing differences. These assumptions are satisfied in most of the applications in business transactions and are assumed
 in the reputation models of Mailath and Samuelson (2001,2015), Ekmekci (2011), Liu (2011), Liu and Skrzypacz (2014), among many others.
I also assume that the short-run players \textit{cannot} directly observe calendar time, which is a standard assumption in reputation models with limited memories such as Liu (2011), Liu and Skrzypacz (2014), and Levine (2021).

This section discusses the two assumptions that distinguish my model from the model of
Liu and Skrzypacz (2014). I examine two alternative models which are only one-step-away from both my baseline model and the model of Liu and Skrzypacz (2014).
Section \ref{sub5.1} studies a model where $u_1$ is \textit{supermodular} but
player 2 \textit{can} observe the exact sequence of player 1's last $K$ actions. Section \ref{sub5.2} studies a model where player 2 \textit{cannot} observe the exact sequence of player $1$'s last $K$ actions but $u_1$ is \textit{submodular}.

\subsection{Observing the Exact Sequence of Player 1's Last $K$ Actions}\label{sub5.1}
I modify the assumption that   player $2$ cannot observe the exact sequence of player $1$'s last $K$ actions while maintaining all other assumptions in my baseline model. My theorems in Section \ref{sec3} extend to the case where each player $2$ knows the exact sequence of player $1$'s last $K$ actions with some small probability $\varepsilon$. The proof of this robustness result is similar to the robustness result under noisy information, which I have discussed in Section \ref{sub4.4} as well as Online Appendix E. I omit the details in order to avoid repetition.

Next, I consider the other extreme case where player $2$ can \textit{perfectly} observe the exact sequence of player 1's last $K$ actions, which is assumed in Liu and Skrzypacz (2014).
First, I show that for every $K \geq 1$, there always \textit{exists} a PBE where actions are $(a^*,b^*)$ in every period and players' payoffs are $u_1(a^*,b^*)$ and $u_2(a^*,b^*)$.
This stands in contrast to the conclusion in Liu and Skrzypacz (2014), that reputation cycles (ones where player $1$ milks his reputation when it is strictly positive and later restore his reputation) occur in all equilibria. This difference explains how the supermodularity or submodularity of player $1$'s stage-game payoff affects the dynamics of his behavior.
Second, I show that in the supermodular product choice game, for all large enough $K$,\footnote{The requirement that $K$ being large enough is needed. To see this, note that when $K=1$, whether player 2 can observe the order of player 1's last $K$ actions is irrelevant, in which case all my results in Section \ref{sec3} apply.}
there also exist equilibria with reputation cycles in which players' payoffs are strictly bounded below
$u_1(a^*,b^*)$ and $u_2(a^*,b^*)$. This implies that when players' actions are strategic complements, allowing the consumers to observe the exact sequence of actions can generate bad equilibria if consumers' memories are long enough. This stands in contrast to games with submodular payoffs studied by Liu and Skrzypacz (2014) in which the seller can secure his commitment payoff when $K$ is large enough.
\begin{Proposition}\label{Prop1}
Suppose $(u_1,u_2)$ satisfies Assumption \ref{Ass1} and inequality (\ref{optimalcommitment}), and that
for every $t \in \mathbb{N}$, player $2_t$ can observe player $1$'s last $\min \{t,K\}$ actions including the exact sequence of these actions.
\begin{enumerate}
    \item For every $K \geq 1$, there exists
    $\underline{\delta} \in (0,1)$ such that when $\delta > \underline{\delta}$, there exists
    a PBE where player $1$ obtains payoff $u_1(a^*,b^*)$ and player $2$ obtains payoff $u_2(a^*,b^*)$.
    \item In the product choice game, there exist $\underline{K} \in \mathbb{N}$, $\overline{\pi} \in (0,1)$, and $\eta>0$ such that when $\pi_0 \in (0,\overline{\pi})$ and $K \geq \underline{K}$,\footnote{The requirement that $\pi_0$ being small enough is necessary for the existence of a low-payoff equilibrium. This is because when $\pi_0$ is large enough, player $2$ has a strict incentive to play $T$ upon observing $(H,...,H)$ given that the commitment type plays $H$ in every period, in which case player $1$ can secure his commitment payoff by playing $H$ in every period.} for every $\delta$ large enough,\footnote{
 As in Liu and Skrzypacz (2014), reputation cycles can occur only if $\delta$ is large enough. This is because
restoring reputation requires player 1 to play the strictly dominated action $H$. Hence, he has no incentive to restore his reputation when $\delta$ is low.}  there exists a PBE where player $1$'s  payoff is no more than $u_1(a^*,b^*)-\eta$ and player $2$'s payoff is no more than $u_2(a^*,b^*)-\eta$.
\end{enumerate}
\end{Proposition}
The proof is in Online Appendix F.
According to this result, when players' actions are strategic complements and the short-run players can observe the order of the patient player's last $K$ actions, there still \textit{exist} equilibria in which the patient player attains his commitment payoff
$u_1(a^*,b^*)$
and the short-run players attain payoff $u_2(a^*,b^*)$. However, observing the order of actions together with a long enough memory can also generate bad equilibria in which both players' payoffs are bounded below $u_1(a^*,b^*)$ and $u_2(a^*,b^*)$.

Comparing this model to Liu and Skrzypacz (2014), the only difference
is that they assume that $u_1$ is submodular while I assume that $u_1$ is supermodular. This leads to two differences in terms of the results. First, they show that \textit{reputation cycles}, ones where player $1$ milks his reputation when it is strictly positive and later restores his reputation, will occur in all equilibria, while I show that there always exists an equilibrium where reputation cycles do not occur. Second, they show that the patient player can secure his commitment payoff in all equilibria when $K$ is large enough. In contrast, I show that there exist equilibria where the patient player's payoff is bounded below his commitment payoff when $K$ is large enough.

The comparison between this result, the no-back-loop lemma in my baseline model, and the result in
Liu and Skrzypacz (2014) implies that whether reputation cycles occur in equilibrium hinges on whether players' actions are complements or substitutes in the stage game. In particular, (i) reputation cycles will inevitably occur when players' actions are substitutes, (ii) reputation cycles may not occur when their actions are complements, and (iii) reputation cycles will never occur when their actions are complements and the short-run players' decisions depend only on the summary statistics of the patient player's recent actions.

Proposition \ref{Prop1} also implies that the patient player's worst equilibrium payoff is \textit{not monotone} with respect to the quality of the short-run players' information, measured in the sense of Blackwell. Theorem \ref{Theorem1} and Fudenberg and Levine (1989)'s  result imply that the patient player can secure his Stackelberg payoff in all equilibria when the short-run players can observe the entire history of his actions \textit{or} when they can only observe the summary statistics of his last $K$ actions. However, the patient player's lowest equilibrium payoff is bounded below his Stackelberg payoff when the short-run players can observe his actions in the last $K (\geq \underline{K})$ periods as well as the order of these actions.
One can also show that when the short-run players can observe the patient player's last $K$ actions including the order of these actions, the patient player's worst equilibrium payoff is weakly decreasing in $K$, in which case a more informative monitoring technology lowers the patient player's worst equilibrium payoff. The proof is available upon request.

The comparison between Proposition \ref{Prop1} and Theorem \ref{Theorem1} also suggests that allowing the myopic uninformed players to observe the exact sequence of the informed player's actions changes the set of equilibrium payoffs. This stands in contrast to repeated Bayesian games where the uninformed player is patient in which Renault, Solan and Vieille (2013) show that it is sufficient for the uninformed player to check the frequency with which the informed player played each of his actions. This is because when the uninformed players are short-lived, their incentives depend only on their beliefs about the informed player's current-period action, while the uninformed player in Renault, Solan, and Vieille (2013) have \textit{intertemporal incentives}.



\subsection{The Patient Player Has Submodular Payoffs}\label{sub5.2}
I relax the assumption that
 $u_1 (a,b)$ has strictly increasing differences while maintaining all other assumptions in my baseline model.
I study the case where $u_1(a,b)$ has \textit{weakly decreasing differences}. The only difference between this model and the one in Liu and Skrzypacz (2014) is that player $2$ knows the exact sequence of player $1$'s last $K$ actions in their model while they do not know that in the current model.

Due to the complications in constructing equilibria, I focus on the product choice game, which is also the primary focus of
Liu (2011) and Liu and Skrzypacz (2014).
The weakly decreasing difference assumption translates into $c_T \geq c_N$. I show that when $\delta$ is large enough,
(i) the no-back-loop lemma fails, i.e., there exists  a best reply of the patient player in which he milks his reputation when it is strictly positive and then restores his reputation, and (ii) if in addition, that the prior probability of commitment type $\pi_0$ is not too large,
there exist equilibria where player $1$'s payoff being bounded below his commitment payoff when $u_1(a,b)$ is \textit{sufficiently submodular}, i.e., $c_T$ is large enough relative to $c_N$.
\begin{Proposition}\label{Prop2}
In the product choice game where $u_1(a,b)$ has weakly decreasing differences.
\begin{enumerate}
    \item For every $K \geq 1$ and $c_T \geq c_N >0$,
    there exists $\underline{\delta} \in (0,1)$ such that for every $\delta > \underline{\delta}$,
    there exist $\sigma_2$ and a pure strategy $\widehat{\sigma}_1$ that best replies to $\sigma_2$ such that the no-back-loop property fails under $(\widehat{\sigma}_1,\sigma_2)$.
    \item Suppose $c_T$ is large enough such that $1+c_T > K (1+c_N)$. There exist
    $\underline{\delta} \in (0,1)$, $\overline{\pi}_0>0$, and
    $\eta>0$ such that for every $\delta > \underline{\delta}$ and $\pi_0 < \overline{\pi}_0$, there exists a PBE where player $1$'s payoff is lower than $u_1(a^*,b^*)-\eta$ and the discounted frequency with which he plays $a^*$ is no more than $1-\eta$.
\end{enumerate}
\end{Proposition}
The comparison between Proposition \ref{Prop2} and the conclusions in my baseline model
implies that whether \textit{reputation cycles} occur in equilibrium hinges on the supermodularity or submodularity of the seller's stage-game payoff function. In particular, a patient seller has an incentive to milk and then rebuild his reputation when his effort and consumers' trust are strategic substitutes, but has no incentive to do so when his effort and consumers' trust are strategic complements.

The proof is in Online Appendix G.
For some intuition, consider the case where $K=1$.
When $c_T  \geq c_N>0$, it is still true that consumer $t$ plays $T$ with strictly higher probability when $a_{t-1}=H$. However, the seller has a stronger incentive to exert high effort when $a_{t-1}=L$. As a result, he may find it optimal to first milk his reputation and then restore his reputation.
This \textit{back loop} that contains the clean history, which is ruled out in the case with supermodular payoffs, provides consumers a rationale for not trusting the seller even after they observe high effort in the period before. My proof  constructs an equilibrium where the seller exerts low effort when $a_{t-1}=H$ and mixes between high and low effort when $a_{t-1}=L$. Consumer $t$ plays $N$ when $a_{t-1}=L$ and plays $T$ with probability between $0$ and $1$ when $a_{t-1}=H$.

\section{Conclusion}\label{sec6}
I analyze a novel reputation model in which the consumers have limited memories and do not know the exact sequence of the seller's actions.
I show that when players' stage-game payoffs are monotone-supermodular, it is never optimal for the seller to milk his reputation and later restore his reputation. This stands in contrast to the conclusion in Liu and Skrzypacz (2014) where reputation cycles I ruled out occur in all equilibria. My main result shows that a sufficiently patient seller receives at least his commitment payoff in all equilibria regardless of consumerss' memory length, which to the best of my knowledge, is the first reputation result that allows for arbitrary memory length. I also show that the patient seller will play his commitment action in almost all periods in all equilibria, and that the consumers can approximately attain their first best welfare in all equilibria \textit{if and only if} the consumers' memory length is lower than some cutoff. The intuition is that although a larger $K$ enables more consumers to punish the seller once the seller shirks, it undermines each consumer's incentive to punish the seller after they observe shirking.

\end{spacing}
\newpage
\appendix

\section{Proof of the No-Back-Loop Lemma}\label{secA}
I provide a unified proof for the no-back-loop lemmas stated in Sections \ref{sub3.2} and \ref{sub4.2}.
For every $t \geq K$, player $2_t$'s incentive depends only on the number of times that player 1 takes each action in the last $K$ periods. Therefore, player 1's continuation value and incentive in period $t$ depend only on $(a_{t-K},...,a_{t-1})$. Although the order of actions in the vector $(a_{t-K},...,a_{t-1})$ does not affect player 2's action, it can affect player 1's incentives. Moreover, player 1's action in period $t$ may depend on variables other than $(a_{t-K},...,a_{t-1})$, such as his actions more than $K$ periods ago and previous player 2's actions.

Fix any $\sigma_2: \mathcal{H}_2 \rightarrow \mathcal{B}^*$. Let $V(a_{t-K},...,a_{t-1})$ be player 1's continuation value in period $t$. Let $\beta^* \in \Delta (B)$ be player 2's action at histories that belong to $\mathcal{H}_1^*$ under $\sigma_2$. For every $a \neq a^*$, let $\beta(a)$ be player 2's action under $\sigma_2$ when exactly one of player 1's last $K$ actions was $a$ and the other $K-1$ actions were $a^*$.
A pure strategy $\widehat{\sigma}_1$ is \textit{canonical} if it depends only on the last $K$ actions of player $1$'s.
For every strategy profile $(\sigma_1,\sigma_2)$ and $h^t \in \mathcal{H}_1$, let $\mathcal{H}_1(\sigma_1,\sigma_2|h^t)$ be the set of histories $h^s$ satisfying $h^s \succ h^t$ and $h^s$ occurring with positive probability when the game starts from history  $h^t$ and players use strategies $(\sigma_1,\sigma_2)$.
If $\widehat{\sigma}_1$ is canonical, then
$\mathcal{H}_1(\widehat{\sigma}_1,\sigma_2|h^t)=\mathcal{H}_1(\widehat{\sigma}_1,\sigma_2'|h^t)$ for every $\sigma_2,\sigma_2'$, and $h^t$.

Since player 2's action depends only on player 1's actions in the last $K$ periods, for every $\sigma_2$, there exists a canonical pure strategy $\widehat{\sigma}_1$ that best replies to $\sigma_2$. Therefore, as long as there exists a pure strategy that best replies to $\sigma_2$ and violates the no-back-loop property with respect to $\sigma_2$, there also exists a canonical pure strategy the best replies to $\sigma_2$ and violates the no-back-loop property with respect to $\sigma_2$. Hence, the no-back-loop lemma is implied by the following \textit{no-back-loop lemma*}, which I show next.
\begin{Lemma3}
 For any $\sigma_2: \mathcal{H}_2 \rightarrow \mathcal{B}^*$ and any canonical pure strategy $\widehat{\sigma}_1$ that best replies to $\sigma_2$. If there exists $h^t \in \mathcal{H}_1^* \bigcap \mathcal{H}_1(\widehat{\sigma}_1,\sigma_2)$ such that $\widehat{\sigma}_1(h^t) \neq a^*$, then $\mathcal{H}_1(\widehat{\sigma}_1,\sigma_2|h^t) \bigcap \mathcal{H}_1^*=\emptyset$.
\end{Lemma3}

Suppose by way of contradiction that there exists a canonical pure strategy $\widehat{\sigma}_1$ that best replies to $\sigma_2$ such that there exist two histories $h^t,h^s \in \mathcal{H}_1^* \bigcap \mathcal{H}_1(\widehat{\sigma}_1,\sigma_2)$
that satisfy $h^s \in \mathcal{H}_1(\widehat{\sigma}_1,\sigma_2|h^t)$, and $\widehat{\sigma}_1(h^t)=a'$ for some $a' \neq a^*$. Without loss of generality, let $h^s$ be the first history in $\mathcal{H}_1^*$ that succeeds $h^t$ when player 1 behaves according to $\widehat{\sigma}_1$. Let $h^{s-1} \equiv (a_0,...,a_{s-2})\in \mathcal{H}_1(\widehat{\sigma}_1,\sigma_2|h^t)$. Since $h^s$ is the first history in $\mathcal{H}_1^*$ that succeeds $h^t$, it must be the case that $h^{s-1} \notin \mathcal{H}_1^*$, so $(a_{s-K-1},...,a_{s-2})=(a'',a^*,...,a^*)$ for some $a'' \neq a^*$. Since $h^s \in \mathcal{H}_1^*$,
player 1 plays $a^*$ at $h^{s-1}$ when he uses strategy $\widehat{\sigma}_1$. This implies that
\begin{equation}\label{3.1}
    (1-\delta) u_1(a^*,\beta(a''))+\delta V(a^*,a^*,...,a^*,a^*)
    \geq (1-\delta) u_1(a',\beta(a''))+\delta V(a^*,a^*,...,a^*,a').
\end{equation}
Since $\widehat{\sigma}_1(h^t)=a'$, player 1 weakly prefers $a'$ to $a^*$ at histories in $\mathcal{H}_1^*$, we have:
\begin{equation}\label{3.2}
    (1-\delta) u_1(a^*,\beta^*)+\delta V(a^*,a^*,...,a^*,a^*)
    \leq (1-\delta) u_1(a',\beta^*)+\delta V(a^*,a^*,...,a^*,a').
\end{equation}
Since the seller's stage-game payoff function is strictly supermodular, and $\beta^*$ and $\beta(a'')$ can be ranked according to FOSD under Assumption 2, inequalities (\ref{3.1}) and (\ref{3.2}) imply that $\beta^* \preceq_{FOSD} \beta(a'')$. Let
\begin{equation}\label{U}
    U \equiv \frac{\sum_{\tau=t+1}^{s-2} \delta^{\tau-t} u_1(\widehat{\sigma}_1(h^{\tau}), \sigma_2(h^{\tau}))}{\sum_{\tau=t+1}^{s-2} \delta^{\tau-t}}.
\end{equation}
be player 1's discounted average payoff from period $t+1$ to period $s-2$ when his period $t$ history is $h^t$ and players play according to $(\widehat{\sigma}_1,\sigma_2)$. Since the strategic-type player 1's incentive depends only on his actions in the last $K$ periods, when $(a_{s-K-1},...,a_{s-2})=(a'',a^*,...,a^*)$, the following strategy is optimal for him:
\begin{itemize}
  \item \textbf{Strategy $*$:} Play $a^*$ in period $s-1$, play $a'$ in period $s$, play $\widehat{\sigma}_1(h^{\tau})$ in period $\tau+(s-t)$ for every $\tau \in \{t+1,...,s-2\}$, and play the same action that he has played $s-t$ periods ago in every period after period $2s-t-1$.
\end{itemize}
Since Strategy $*$ is optimal for player 1, it must yield a weakly greater payoff compared to any of the following two deviations starting from a period $s-1$ history where $(a_{s-K-1},...,a_{s-2})=(a'',a^*,...,a^*)$:
\begin{itemize}
  \item \textbf{Deviation A:} Play $a'$ in period $s-1$, $\widehat{\sigma}_1(h^{\tau})$ in period $\tau+(s-t-1)$ for every $\tau \in \{t+1,...,s-2\}$, and play the same action that he has played $s-t-1$ periods ago in every period after $2s-t-2$.
  \item \textbf{Deviation B:} Play $a''$ in period $s-1$, play $a^*$ from period $s$ to $s+K-2$, and play the same action that he has played $K$ periods ago in every period after $s+K-1$.
\end{itemize}
Player 1 prefers Strategy $*$ to Deviation A, which implies that:
\begin{equation*}
    \frac{(1-\delta) u_1(a',\beta(a''))+(\delta-\delta^{s-t-2}) U }{1-\delta^{s-t-2}}
    \leq \frac{(1-\delta) u_1(a^*,\beta(a''))+(1-\delta)\delta u_1(a',\beta^*)+(\delta^2-\delta^{s-t-1}) U }{1-\delta^{s-t-1}}
\end{equation*}
This leads to the following upper bound on $U$:
\begin{equation}\label{3.3}
    (\delta-\delta^{s-t-2}) U \leq (1-\delta^{s-t-2}) u_1(a^*,\beta(a''))
    +\delta(1-\delta^{s-t-2}) u_1(a',\beta^*)-(1-\delta^{s-t-1}) u_1(a',\beta(a'')).
\end{equation}
Player 1 prefers Strategy $*$ to Deviation B, which implies that:
\begin{equation}\label{3.4}
    \frac{(1-\delta) u_1(a'',\beta(a''))+(\delta-\delta^{K}) u_1(a^*,\beta(a'')) }{1-\delta^{K}}
    \leq \frac{(1-\delta) u_1(a^*,\beta(a''))+(1-\delta)\delta u_1(a',\beta^*)+(\delta^2-\delta^{s-t-1}) U }{1-\delta^{s-t-1}}.
\end{equation}
This leads to a lower bound on $U$. The left-hand-side of (\ref{3.4}) equals
\begin{equation*}
  u_1(a^*,\beta(a'')) +  \frac{1-\delta}{1-\delta^K} \underbrace{\Big\{
    u_1(a'',\beta(a''))-u_1(a^*,\beta(a''))
    \Big\}}_{>0, \textrm{ since } a'' \prec a^* \textrm{ and } u_1 \textrm{ is decreasing in } a},
\end{equation*}
and inequality (\ref{3.3}) implies that the right-hand-side of (\ref{3.4}) is no more than:
\begin{equation*}
    u_1(a^*,\beta(a'')) + \delta\underbrace{ \Big\{
    u_1(a',\beta^*)-u_1(a',\beta(a''))
    \Big\}}_{ \leq 0, \textrm{ since } \beta(a'') \succeq \beta^* \textrm{ and } u_1 \textrm{ is increasing in } b}.
\end{equation*}
Since $u_1(a,b)$ is strictly increasing in $b$ and is strictly decreasing in $a$, $a^* \succ a''$, and $\beta(a'') \succeq \beta^*$, inequality (\ref{3.4}) cannot be true. This leads to a contradiction and implies the no-back-loop lemma.

\section{Proof of Theorems 2 and 3}\label{secB}
In Section \ref{subB.1}, I complete the proof in Section \ref{sub3.3} and show that when $K < \overline{K}$, the patient player plays $a^*$ with frequency arbitrarily close to $1$ in all equilibria. I focus on the case where (\ref{optimalcommitment}) is satisfied since when (\ref{optimalcommitment}) is not satisfied, $\overline{K}=1$ and there is no $K$ that meets the requirement of Statement 1, i.e., the statement is trivially satisfied.
In Section \ref{subB.2}, I construct equilibria where the patient player plays $a^*$ with frequency bounded away from $1$ when (\ref{optimalcommitment}) is satisfied and $K \geq \overline{K}$, as well as equilibria where the patient player plays $a^*$ with zero frequency when (\ref{optimalcommitment}) is violated.


\subsection{The Frequency of $a^*$ Being Close to $1$ in All Equilibria}\label{subB.1}
I start from the proof of Lemma \ref{L4.1}.
\begin{proof}[Proof of Lemma 3.1:]
For every $t \in \mathbb{N}$, let $p_t(s)$ be the probability that the state is $s$ in period $K+t$ conditional on player 1 being the strategic type,
and let $q_t(s \rightarrow s')$ be the probability that the state in period $t+K+1$ is $s'$ conditional on the state being $s$ in period $t+K$ and player 1 being the strategic type. By definition, $p_0(s)=p(s)$ and $p_{t+1}(s) = \sum_{s' \in S} p_t(s') Q(s' \rightarrow s)$.
According to Bayes rule, we have
\begin{equation*}
    \mu(s)= \sum_{t=0}^{+\infty} (1-\delta)\delta^t p_t(s)
\quad \textrm{and} \quad
    Q(s \rightarrow s') = \frac{\sum_{t=0}^{+\infty} (1-\delta)\delta^t p_t(s) q_t(s \rightarrow s')}{\sum_{t=0}^{+\infty} (1-\delta)\delta^t p_t(s)}.
\end{equation*}
This implies that
\begin{equation*}
     \sum_{s \in S} \mu(s) Q(s \rightarrow s') = \sum_{s \in S} \sum_{t=0}^{+\infty} (1-\delta)\delta^t p_t(s) q_t(s \rightarrow s')
     = \sum_{t=0}^{+\infty} (1-\delta)\delta^t \sum_{s \in S} p_t(s) q_t(s \rightarrow s')=\sum_{t=0}^{+\infty} (1-\delta)\delta^t p_{t+1}(s')
\end{equation*}
\begin{equation*}
   \frac{1}{\delta} \Big\{
    \mu(s')-(1-\delta) p(s')
    \Big\}=  \frac{1}{\delta} \Big\{
 \sum_{t=0}^{+\infty} (1-\delta)\delta^t p_t(s')  -(1-\delta) p(s')
    \Big\}= \sum_{t=0}^{+\infty} (1-\delta)\delta^t p_{t+1}(s').
\end{equation*}
These two equations together imply (\ref{4.8}).
\end{proof}
Next, I show Lemma \ref{L4.3}.
\begin{proof}[Proof of Lemma 3.3:] Since $\mathcal{Q} (S_{j,k} \rightarrow S_{k-1})= \mathcal{Q}(S_{j,k}' \rightarrow S_{k-1})$, $\mathcal{Q}(S_{j,k} \rightarrow S_{k+1})=\mathcal{Q}(S_{j,k}^* \rightarrow S_{k+1})$, and under the hypothesis that $\mathcal{Q} (S_{j,k} \rightarrow S_{k-1})\leq z(1-\delta)$ and $\mathcal{Q} (S_{k-1} \rightarrow S_{j,k}) \leq z(1-\delta)$, we have:
\begin{equation*}
    \underbrace{\sum_{s \in S_{j,k}^*} \sum_{s' \in S_{j,k}} \mu(s) Q(s \rightarrow s')}_{= \sum_{s \in S_{j,k}^*} \mu(s) - \mathcal{Q}(S_{j,k}^* \rightarrow S_{k+1})}+\underbrace{\mathcal{Q} (S_{j,k} \rightarrow S_{k-1})}_{\leq z(1-\delta)}
    \leq \sum_{s \in S_{j,k}^*} \mu(s) - \mathcal{Q}(S_{j,k}^* \rightarrow S_{k+1})+z (1-\delta),
\end{equation*}
\begin{equation*}
    \underbrace{\sum_{s \in S_{j,k}'} \sum_{s' \in S_k} \mu(s) Q(s \rightarrow s')}_{=\sum_{s \in S_{j,k}'} \mu(s)-\mathcal{Q}(S_{j,k}' \rightarrow S_{k-1})} +\underbrace{\mathcal{Q}(S_{j,k} \rightarrow S_{k+1})}_{=\mathcal{Q}(S_{j,k}^* \rightarrow S_{k+1})}
    \geq \sum_{s \in S_{j,k}'} \mu(s)  + \mathcal{Q}(S_{j,k}^* \rightarrow S_{k+1})-z (1-\delta).
\end{equation*}
Suppose there exists no such $y \in \mathbb{R}_+$, that is, $\frac{\sum_{s \in S_{j,k}} \mu(s)}{z(1-\delta)}$ can be arbitrarily large as $\delta \rightarrow 1$. Since
the sum of $\sum_{s \in S_{j,k}^*} \mu(s) - \mathcal{Q}(S_{j,k}^* \rightarrow S_{k+1})+z (1-\delta)$ and $\sum_{s \in S_{j,k}'} \mu(s)  + \mathcal{Q}(S_{j,k}^* \rightarrow S_{k+1})-z (1-\delta)$ equals $\sum_{s \in S_{j,k}} \mu(s)$, we know that when $\delta$ is close to $1$, (\ref{incentive}) is implied by:
\begin{equation*}
    \frac{\sum_{s \in S_{j,k}^*} \mu(s) - \mathcal{Q}(S_{j,k}^* \rightarrow S_{k+1})}{\sum_{s \in S_{j,k}'} \mu(s)  + \mathcal{Q}(S_{j,k}^* \rightarrow S_{k+1})}<K-1,
\end{equation*}
or equivalently,
\begin{equation}\label{incentives}
    \sum_{s \in S_{j,k}^*} \mu(s) < (K-1) \sum_{s \in S_{j,k}'} \mu(s)+ K \mathcal{Q}(S_{j,k}^* \rightarrow S_{k+1}).
\end{equation}
I derive a lower bound for $\mathcal{Q}(S_{j,k}^* \rightarrow S_{k+1})$.
Since $\mathcal{O}(S_{j,k}) = \mathcal{Q}(S_{j,k}^* \rightarrow S_{k+1}) + \mathcal{Q}(S_{j,k}' \rightarrow S_{k-1}) + \mathcal{Q}(S_{j,k} \rightarrow S_k \backslash S_{j,k})$,
under the hypothesis that  $\max \{ \mathcal{Q} (S_{j,k} \rightarrow S_{k-1}),\mathcal{Q} (S_{k-1} \rightarrow S_{j,k}) \} \leq z(1-\delta)$,
\begin{equation*}
    \mathcal{Q}(S_{j,k}^* \rightarrow S_{k+1})=
    \mathcal{O}(S_{j,k})- \mathcal{Q}(S_{j,k}' \rightarrow S_{k-1}) - \mathcal{Q}(S_{j,k} \rightarrow S_k \backslash S_{j,k})
    \geq
    \mathcal{O}(S_{j,k})- \mathcal{Q}(S_{j,k} \rightarrow S_k \backslash S_{j,k}) -z(1-\delta).
\end{equation*}
According to Lemma \ref{Linflow}, we have
 $\mathcal{Q}(S_{j,k}^* \rightarrow S_{k+1}) \geq \mathcal{I}(S_{j,k})- \mathcal{Q}(S_{j,k} \rightarrow S_k \backslash S_{j,k})- \frac{(1-\delta)(1+z\delta)}{\delta}$.
Since at every $s \in S_{j,k}^*$, the state in the next period belongs to $S_{k+1}$ if player 1 does not play $a^*$ at $s$, and belongs to $S_{j,k}$ if player 1 plays $a^*$ at $s$, we have $\mathcal{Q}(S_{j,k}^* \rightarrow S_k \backslash S_{j,k})=0$. This implies that $\mathcal{Q}(S_{j,k} \rightarrow S_k \backslash S_{j,k})=\mathcal{Q}(S_{j,k}' \rightarrow S_k \backslash S_{j,k})$. Since $\mathcal{I}(S_{j,k})=\mathcal{I}(S_{j,k}^*) +  \mathcal{I}(S_{j,k}')
- \mathcal{Q}(S_{j,k}^* \rightarrow S_{j,k}')-\mathcal{Q}(S_{j,k}' \rightarrow S_{j,k}^*)$,
\begin{eqnarray*}
\mathcal{Q}(S_{j,k}^* \rightarrow S_{k+1}) & \geq & \mathcal{I}(S_{j,k})- \mathcal{Q}(S_{j,k} \rightarrow S_k \backslash S_{j,k})- \frac{(1-\delta)(1+z\delta)}{\delta}
{}
\nonumber\\
&=& {}
\mathcal{I}(S_{j,k}^*) +  \mathcal{I}(S_{j,k}')
- \mathcal{Q}(S_{j,k}^* \rightarrow S_{j,k}')-\mathcal{Q}(S_{j,k}' \rightarrow S_{j,k}^*)
- \mathcal{Q}(S_{j,k}' \rightarrow S_k \backslash S_{j,k})- \frac{(1-\delta)(1+z\delta)}{\delta}
{}
\nonumber\\
&\geq & {} \mathcal{I}(S_{j,k}^*) -\mathcal{Q}(S_{j,k}^* \rightarrow S_{j,k}') +  \underbrace{\mathcal{O}(S_{j,k}')- \mathcal{Q}(S_{j,k}' \rightarrow S_k \backslash S_{j,k}) -\mathcal{Q}(S_{j,k}' \rightarrow S_{j,k}^*)}_{\geq 0}- \frac{(1-\delta)(2+z\delta)}{\delta}
{}
\nonumber\\
&\geq & {}
\mathcal{I}(S_{j,k}^*) -\mathcal{Q}(S_{j,k}^* \rightarrow S_{j,k}')- \frac{(1-\delta)(2+z\delta)}{\delta}
\end{eqnarray*}
Since $\mathcal{Q}(S_{j,k}^* \rightarrow S_{j,k}') \leq \mathcal{O}(S_{j,k}^*)$ and $\mathcal{Q}(S_{j,k}^* \rightarrow S_{j,k}') \leq \mathcal{I}(S_{j,k}')$,
\begin{equation*}
    \mathcal{Q}(S_{j,k}^* \rightarrow S_{j,k}') \leq \frac{1}{K} \mathcal{O}(S_{j,k}^*)+\frac{K-1}{K} \mathcal{I}(S_{j,k}')
    \leq \frac{1}{K} \mathcal{I}(S_{j,k}^*)+\frac{K-1}{K} \mathcal{O}(S_{j,k}')+\frac{1-\delta}{\delta}.
\end{equation*}
This together with the lower bound on $\mathcal{Q}(S_{j,k}^* \rightarrow S_{k+1})$ that we derived earlier implies that:
\begin{equation*}
    \mathcal{Q}(S_{j,k}^* \rightarrow S_{k+1}) \geq \frac{K-1}{K} \Big(
\mathcal{I}(S_{j,k}^*)-\mathcal{O}(S_{j,k}')
\Big)-\frac{(1-\delta)(3+z\delta)}{\delta}.
\end{equation*}
Since $\mathcal{O}(S_{j,k}')=\mathcal{Q}(S_{j,k}' \rightarrow S_k \backslash S_{j,k})+\mathcal{Q}(S_{j,k}' \rightarrow S_{j,k}^*)+\mathcal{Q}(S_{j,k}' \rightarrow S_{k-1})$, and $\mathcal{Q}(S_{j,k}' \rightarrow S_{k-1})$ is assumed to be less than $z(1-\delta)$, we know that
\begin{equation*}
    \mathcal{Q}(S_{j,k}^* \rightarrow S_{k+1}) \geq \frac{K-1}{K} \Big(
\mathcal{I}(S_{j,k}^*)-\mathcal{Q}(S_{j,k}' \rightarrow S_k \backslash S_{j,k})-\mathcal{Q}(S_{j,k}' \rightarrow S_{j,k}^*)
\Big)-\frac{(1-\delta)(3+2z\delta)}{\delta}.
\end{equation*}
Hence, when $\delta$ is close to $1$, inequality (\ref{incentives}) is implied by
\begin{eqnarray}\label{incentives1}
\sum_{s \in S_{j,k}^*} \mu(s) &<&
(K-1)\Big\{ \sum_{s \in S_{j,k}'} \mu(s) - \mathcal{Q}(S_{j,k}' \rightarrow S_k \backslash S_{j,k}) \Big\}
{}
\nonumber\\
&  & {}
+ (K-1) \Big\{\mathcal{I}(S_{j,k}^*) - \mathcal{Q}(S_{j,k}' \rightarrow S_{j,k}^*)\Big\}.
\end{eqnarray}

I say that a sequence of states $\{s_1,...,s_j\} \subset S_{j,k}$ form a \textit{connected sequence} if for every $i \in \{1,2,...,i-1\}$, there exists $a_i \in A$  such that the state in the next period is $s_{i+1}$ when the state in the current period is $s_i$ and player 1 takes action $a_i$. 
A useful observation is that for every $s \in S_{j,k}$, there exists a unique state $s'$ in $S_{j,k}$ such that playing some $a \in A$ in state $s'$ leads to state $s$. Using this observation, we construct for any $s_1 \in S_{j,k}'$,
a finite sequence of states $\{s_1,...,s_m\} \subset S_{j,k}$ with length $m$ at least one such that (i) for every $i \in \{1,2,...,m-1\}$, there exists an action $a_i \in A$ such that playing $a_i$ in state $s_i$ leads to state $s_{i+1}$ in the next period, (ii) if $m \geq 2$, then $\{s_2,...,s_m\} \subset S_{j,k}^*$, and (iii) no matter which action player 1 takes in state $s_m$, the state in the next period does not belong to $S_{j,k}^*$, or equivalently, there exists an action $a \in A$ such that taking action $a$ at state $s_m$ leads to a state that belongs to $S_{j,k}'$.
Lemma \ref{L4.1} implies that
\begin{eqnarray}\label{4.19a}
\mu(s_2) &\leq& \mu(s_1) Q(s_1 \rightarrow s_2)+ \mathcal{Q}(S \backslash S_{j,k} \rightarrow \{s_2\})
    +\frac{1-\delta}{\delta}
{}
\nonumber\\
&=& {} \mu(s_1)-\mathcal{Q}(\{s_1\} \rightarrow S_k \backslash S_{j,k}) + \mathcal{Q}(S \backslash S_{j,k} \rightarrow \{s_2\}) +\frac{1-\delta}{\delta},
\end{eqnarray}
and for every $i \geq 2$, we have:
\begin{eqnarray}\label{4.19}
 \mu(s_{i+1}) &\leq& \mu(s_i) Q(s_i \rightarrow s_{i+1})+ \mathcal{Q}(S \backslash S_{j,k} \rightarrow \{s_{i+1}\})
    +\frac{1-\delta}{\delta}
{}
\nonumber\\
&=& {} \mu(s_i)+\mathcal{Q}(S \backslash S_{j,k} \rightarrow \{s_{i+1}\})
    +\frac{1-\delta}{\delta}.
\end{eqnarray}
Iteratively apply (\ref{4.19}) and (\ref{4.19a}) for every $i \geq 2$, we obtain:
\begin{equation}\label{4.20}
    \mu(s_i) \leq \mu(s_1) + \mathcal{Q}(S \backslash S_{j,k} \rightarrow \{s_2,...,s_i\}) -\mathcal{Q}(\{s_1\} \rightarrow S_k \backslash S_{j,k})
    +\frac{(1-\delta)(i-1)}{\delta}.
\end{equation}
Summing up inequality (\ref{4.20}) for $i \in \{2,...,m\}$, we obtain:
\begin{eqnarray*}
    \sum_{i=2}^m \mu(s_i)
    &\leq& (m-1) \Big\{ \mu(s_1) + \mathcal{Q}(S_{j,k}^c \rightarrow \{s_2,...,s_m\})
    - \mathcal{Q}(\{s_1\} \rightarrow S_k \backslash S_{j,k})
    \Big\} + \frac{m(m-1)(1-\delta)}{2\delta}
        {}
\nonumber\\
&= & {} (m-1) \Big\{
\underbrace{\mu(s_1)-\mathcal{Q}(\{s_1\} \rightarrow S_k \backslash S_{j,k})}_{\geq -\frac{1-\delta}{\delta}}
\Big\}
        {}
\nonumber\\
& &
+(m-1) \Big\{
\underbrace{\mathcal{Q}(S \backslash S_{j,k}^* \rightarrow \{s_2,...,s_m\})
-\mathcal{Q}(S_{j,k}' \rightarrow \{s_2,...,s_m\})}_{\geq 0}
\Big\}+ \frac{m(m-1)(1-\delta)}{2\delta}
        {}
\nonumber\\
&\leq & {} (K-1) \Big\{
\mu(s_1)-\mathcal{Q}(\{s_1\} \rightarrow S_k \backslash S_{j,k})
\Big\}
        {}
\nonumber\\
& &
      +(K-1) \Big\{
\mathcal{Q}(S \backslash S_{j,k}^* \rightarrow \{s_2,...,s_m\})
-\mathcal{Q}(S_{j,k}' \rightarrow \{s_2,...,s_m\}) \Big\}
+ \Big\{ \frac{m(m-1)(1-\delta)}{2\delta}+\frac{K(1-\delta)}{\delta} \Big\}.
\end{eqnarray*}
One can obtain (\ref{incentives1}) by summing up the above equation for every $s_1 \in S_{j,k}'$ and taking $\delta \rightarrow 1$. This is because the left-hand-side of this sum equals $\sum_{s \in S_{j,k}^*}\mu(s)$. Therefore, after ignoring the last term that vanishes to $0$ as $\delta \rightarrow 1$, the additive property of the operator $\mathcal{Q}$ implies that the right-hand-side equals
\begin{equation*}
    (K-1)\Big\{ \sum_{s \in S_{j,k}'} \mu(s) - \mathcal{Q}(S_{j,k}' \rightarrow S_k \backslash S_{j,k}) \Big\}
+ (K-1) \Big\{\mathcal{I}(S_{j,k}^*) - \mathcal{Q}(S_{j,k}' \rightarrow S_{j,k}^*)\Big\}.
\end{equation*}
This establishes inequality (\ref{incentives1}) and leads to the conclusion of Lemma 3.3.
\end{proof}
Then I show Lemma \ref{L4.4}.
\begin{proof}[Proof of Lemma 3.5:]
Suppose by way of contradiction that for every $y>0$ and $\underline{\delta} \in (0,1)$, there exist $\delta > \underline{\delta}$, an equilibrium under $\delta$, and $k \geq 1$, such that in this equilibrium, $\max \{ \mathcal{Q}(S_{k-1} \rightarrow S_{j,k}), \mathcal{Q}(S_{k} \rightarrow S_{k-1}) \} < y(1-\delta)$ but $\sum_{s \in S_{k}} \mu(s) >  z(1-\delta)$.
Pick a large enough $z$,
Lemma \ref{L4.3} implies that for every $S_{j,k} \subset S_k$, either $\sum_{s \in S_{j,k}} \mu(s)< \frac{z}{2^K}(1-\delta)$, or player 2 has a strict incentive not to play $a^*$ at $S_{j,k}$.
The hypothesis that $\sum_{s \in S_{k}} \mu(s) >  z(1-\delta)$ implies that there exists at least one partition element $S_{j,k}$ such that player 2 has a strict incentive not to play $a^*$ at $S_{j,k}$. Let $S_k'$ be the union of such partition elements.

I start from deriving an upper bound on the ratio between $\sum_{s \in S_{k}'} \mu(s)$ and $\mathcal{Q}(S_{k}' \rightarrow S_{k-1})$.
Let $V(s)$ be player 1's continuation value in state $s$ and let $\overline{V} \equiv \max_{s \in S} V(s)$.
Let $\underline{v}$ be player 1's lowest stage-game payoff. Let $v' \equiv \max_{a \in A, b \prec b^*} u_1(a,b)$ and $v^* \equiv u_1(a^*,b^*)$. Assumptions \ref{Ass1} and inequality (\ref{optimalcommitment}) together imply that $v^*>v'>\underline{v}$.
Since player 1 can reach any state within $K$ periods, we have $V(s) \geq (1-\delta^K) \underline{v} +\delta^K \overline{V}$ for every $s \in S$.
Theorem \ref{Theorem1} suggests that player 1's continuation value at $s^*$ is at least $u_1(a^*,b^*)$. Therefore, $\overline{V} \geq v^*$.
Let $M$ be the largest integer $m$ such that
\begin{equation}\label{B.12}
(1-\delta^m) v' +  \delta^m \overline{V} \geq (1-\delta^K) \underline{v} +\delta^K \overline{V}.
\end{equation}
Applying the L'Hospital Rule, (\ref{B.12}) implies that when $\delta$ is close to $1$, we have
$M \leq K \frac{\overline{V}-\underline{v}}{\overline{V}-v'}$.
Therefore, for any $t \in \mathbb{N}$ and $s \in S_{k}'$, and under any pure-strategy best reply of player 1, if the state is $s$ in period $t$, then there exists $\tau \in \{t+1,...,t+M\}$ such that when player 1 uses this pure-strategy best reply, the state in period $\tau$ does not belong to $S_{k}'$. Therefore,
    $\frac{\sum_{s \in S_{k}'} \mu(s)}{\mathcal{O}(S_{k}')} \leq \frac{1-\delta^M}{\delta^M(1-\delta)}$.
When $\delta \rightarrow 1$, the RHS of the above inequality converges to $M$, which implies that
\begin{equation}\label{B.13}
    \sum_{s \in S_k'} \mu(s) \leq K \cdot \frac{\overline{V}-\underline{v}}{\overline{V}-v'} \cdot \mathcal{O}(S_k').
\end{equation}
Since $\sum_{s \in S_k \backslash S_k'} \mu(s)$ is bounded above by some linear function of $1-\delta$, it must be the case that
$\mathcal{ Q}(S_{k-1} \rightarrow S_k')  \geq \frac{\sum_{s \in S_k'} \mu(s)}{2M}$.
This implies that there exists $s \in S_{k+1}$ and a canonical pure best reply $\widehat{\sigma}_1$ such that:
\begin{enumerate}
  \item the state in the next period, denoted by $s'$, belongs to $S_k'$, and the state belongs to $S_k'$ for $m$ periods,
  \item the state returns to $S_{k+1}$ after these $m$ periods, returns to $s$ after a finite number of periods, and the state never reaches $\cup_{n=0}^{k-1} S_n$ when play starts from $s$.
\end{enumerate}
By definition, player 1 plays $a^*$ in state $s$ under $\widehat{\sigma}_1$ and
$\widehat{\sigma}_1$ induces a cycle of states.
Moreover, it is without loss of generality to focus on best replies that induce a cycle where each state occurs at most once.

I show that $m \leq K-1$. Suppose by way of contradiction that $m \geq K$, namely, after reaching state $s'$, the state belongs to $S_k'$ for at least $K$ periods under player 1's pure-strategy best reply $\widehat{\sigma}_1$. Recall the definition of a \textit{minimal connected sequence}. Every minimal connected sequence contains either one state (if $k=K$) or $K$ states in category $k$. Therefore, the category $k$ state after $K$ periods is also $s'$. As a result, there exists a best-reply of player 1 such that under this best reply and starting from state $s'$, the state remains in category $k$ forever. Due to the hypothesis that player 2 has no incentive to play $b^*$ when the state belongs to $S_k'$, player 1's continuation value under such a best reply is at most $v'$, which is strictly less than his guaranteed continuation value $(1-\delta^K) \underline{v} +\delta^K v^*$. This contradicts the conclusion of Theorem \ref{Theorem1}.

Given that $m \leq K-1$, let us consider an alternative strategy of player 1 under which he plays an action other than $a^*$ in state $s$,
then follows strategy $\widehat{\sigma}_1$. Starting from state $s$, this strategy and $\widehat{\sigma}_1$ lead to the same state after $m+1$ periods. This strategy leads to a strictly higher payoff since the stage-game payoff at state $s$ is strictly greater, and the payoffs after the first period are weakly greater. This contradicts the hypothesis that $\widehat{\sigma}_1$ is player 1's best reply to player 2's equilibrium strategy.
\end{proof}
In summary, Lemma \ref{L4.2} implies that $\max\{ \mathcal{Q}(S_0 \rightarrow S_1),\mathcal{Q}(S_1 \rightarrow S_0)\} \leq \frac{2(1-\delta)}{\delta}$. Lemma \ref{L4.3} and Lemma \ref{L4.4} together imply that $\sum_{s \in S_1} \mu(s)$ is bounded from above by a linear function of $1-\delta$ given that
 $\max\{ \mathcal{I}(S_0), \mathcal{O}(S_0)\} \leq \frac{2(1-\delta)}{\delta}$, which then implies that $\mathcal{Q}(S_1 \rightarrow S_2)$ and $\mathcal{Q}(S_2 \rightarrow S_1)$ are also bounded from above by a linear function of $1-\delta$.
 Iteratively apply this argument, we obtain that for every $k \in \{1,2,...,K\}$, $\sum_{s \in S_k} \mu(s)$
 is bounded from above by a linear function of $1-\delta$.

\subsection{Constructing Equilibria where $a^*$ Occurs with Low Frequency}\label{subB.2}
\paragraph{Case 1:} I consider the case where $(u_1,u_2)$ satisfies (\ref{optimalcommitment}) but $K \geq \overline{K}$.
Since $K \geq \overline{K}$, $b^*$ best replies to $\frac{K-1}{K} a^* + \frac{1}{K} a'$ for some $a' \neq a^*$. Inequality (\ref{optimalcommitment}) implies that every best reply to $a'$ is strictly lower than $b^*$. Hence, $K \geq 2$ and there exists $\alpha \in (0,\frac{K-1}{K})$ such that $\{b^*,b'\} \subset \textrm{BR}_2(\alpha a^*+(1-\alpha)a')$ for some $b' \prec_B b^*$. Let $b''$ be player 2's lowest best reply to $a'$. Since $u_2(a,b)$ has strictly increasing differences, we know that $b'' \preceq_B b' \prec_B b^*$.
Let $\mathcal{H}_1^{**}$ be the set of histories such that player 2 observes at most one $a'$ and does not observe any action other than $a^*$ and $a'$, which will contain the set of histories that occur with positive probability.
Player 2's belief is derived from Bayes rule at every history that belongs to $\mathcal{H}_1^{**}$. For player 2's belief at histories that occur with zero probability,
\begin{enumerate}
\item If player $2_t$ observes two or more $a'$ and observes no action other than $a^*$ and $a'$, then she believes that $(a_{t-2},a_{t-1})=(a',a')$.
\item If player 2 observes $a'' \notin \{a^*,a'\}$, then she believes that the action in the period before is $a''$.
\end{enumerate}

Then I describe player 1's equilibrium strategy. At every history that belongs to $\mathcal{H}_1^{**}$, player 1 plays $a'$ in period $t$ if $t=K-1$ or $t \geq K$ and $(a_{\min\{0,t-K+1\}},...a_{t-1})=(a^*,...,a^*)$. Player 1 plays $a^*$ in period $t$ at other histories that belong to $\mathcal{H}_1^{**}$.
Histories that do not belong to $\mathcal{H}_1^{**}$ occur with zero probability, at which player 1's behavior is given by:
\begin{enumerate}
    \item Player 1 plays $a^*$ if $(a_{t-2},a_{t-1})\neq (a',a')$ and the last $\min \{K,t\}$ actions are either $a^*$ or $a'$.
  \item Player 1 plays $a^*$ with probability $\alpha$ and plays $a'$ with probability $1-\alpha$ if $(a_{t-2},a_{t-1})=(a',a')$ and actions
  in the last $\min \{K,t\}$ periods are either $a^*$ or $a'$.
  \item Player 1 plays $a'$ in period $t$ if actions other than $a^*$ and $a'$ occurred in period $t-1$.
\end{enumerate}

Player $2$ plays $b^*$ at every history that belongs to
$\mathcal{H}_1^{**}$. At every history that (i) does not belong to $\mathcal{H}_1^{**}$, and (ii) actions other than $a^*$ and $a'$ do not occur in the last $K$ periods, player $2$ plays $b^*$ with probability $\beta$ and plays $b'$ with probability $1-\beta$, where
\begin{equation*}
      \beta u_1(a',b^*) +(1-\beta) u_1(a',b')
        = (1-\delta^{K-1}) \Big(\beta u_1(a^*,b^*) +(1-\beta) u_1(a^*,b')\Big)
        \end{equation*}
            \begin{equation}\label{4.3}
        + \delta^{K-1} \underbrace{\frac{u_1(a',b^*)+ (\delta+\delta^2+...+\delta^{K-1})  u_1(a^*,b^*)}{1+\delta+...+\delta^{K-1}}}_{\equiv V_K}.
    \end{equation}
Since $u_1(a',b^*)> u_1(a^*,b^*)> u_1(a',b')> u_1(a^*,b')$,
 $\beta$ is strictly between $0$ and $1$. At every history that does not belong to $\mathcal{H}_1^{**}$ and actions other than $a^*$ and $a'$ occurred in the last $K$ periods, player 2 plays $b''$.

Player 2's incentive constraint  at every history that occurs with zero probability is satisfied under her belief since (i) she mixes between $b^*$ and $b'$ whenever she believes that player 1 plays $\alpha a^* +(1-\alpha) a'$, and (ii) she plays $b''$ whenever she believes that player 1 plays $a'$.
At  histories that occur with positive probability, player 2 believes that
$a'$ is played with probability $1-\pi_0$ and $a^*$ is played with probability $\pi_0$ in period $0$, so she plays a best reply to this mixed action.
From period $1$ to $K-1$, player 2 believes that $a^*$ is played by both types, so she plays her best reply $b^*$.
After period $K$, player 2's belief
assigns probability $1$ to the commitment type upon observing any history where player 1's last $K$ actions were $a^*$, and therefore, she has a strict incentive to play $b^*$.
For player 2's incentive constraints at histories where $a'$ occurred only once, she believes that $(a_{t-K},...,a_{t-1})=(a',a^*,...,a^*)$
with probability
\begin{equation}\label{4.4}
    \frac{\delta^{K-1}}{1+\delta+...+\delta^{K-1}}
\end{equation}
and $(a_{t-K},...,a_{t-1}) \neq (a',a^*,...,a^*)$ with complementary probability. When $\delta \rightarrow 1$, expression (\ref{4.4}) is less than but converges to $\frac{1}{K}$.
Since player 1 plays $a'$ when $(a_{t-K},...,a_{t-1})=(a',a^*,...,a^*)$ and plays $a^*$ at other histories where $a'$ occurred once, player 2 believes that player 1's current period action is $a'$ with probability less than $\frac{1}{K}$ and is $a^*$ with probability more than $\frac{K-1}{K}$.
Hence, there exists $\underline{\delta} \in (0,1)$ such that when $\delta > \underline{\delta}$, player 2s have a strict incentive to play $b^*$ if $a'$ occurred only once in the last $K$ periods.

I verify player 1's incentive constraint: (i) he has no incentive to reach any history that occurs with zero probability starting from any history that occurs with positive probability, and (ii) he has an incentive to play $a'$ when $(a_{t-2},a_{t-1})=(a',a')$ or when $a_{t-1} \notin \{a',a^*\}$.
When $(a_{t-2},a_{t-1})=(a',a')$, player 1 is indifferent between playing $a'$ and $a^*$ in period $t$.

Next, I show that player 1 has no incentive to play $a'$ at every history that belongs to $\mathcal{H}_1^{**}$ where $(a_{t-K+1},...,a_{t-1}) \neq (a^*,...,a^*)$. For every $m \in \{1,2,...,K\}$, let $s_m$ be the state where $a_{t-m}=a'$ and all actions that belong to $\{a_{t-K},..,a_{t-1}\}\backslash \{a_{t-m}\}$
are $a^*$. Let $V_m$ player $1$'s continuation value in state $s_m$.
For every $m \in \{1,2,...,K-1\}$, player $1$ prefers $a^*$ to $a'$ in state $s_m$ if
\begin{equation}\label{4.23}
    V_m
    > (1-\delta) u_1(a',b^*)+ \delta (1-\delta^{K-m}) u_1(a^*,\beta) + (\delta^{K-m+1}-\delta^K) u_1(a^*,b^*) + \delta^K V_K.
\end{equation}
Since
\begin{equation*}
    V_m= (1-\delta)^{K-m} u_1(a^*,b^*) +\delta^{K-m} V_K \textrm{ for every } 1 \leq m \leq K,
\end{equation*}
we have:
\begin{equation*}
        (1-\delta) (1-\delta^{K-m}) u_1(a^*,b^*)
    +\delta (1-\delta^{K-m}) \Big(
    u_1(a^*,b^*)-u_1(a^*,\beta)
    \Big)
\end{equation*}
\begin{equation}\label{4.24}
-(1-\delta) u_1(a',b^*)
    +\delta^{K-m} (1-\delta^m) V_k -(\delta^{K-m+1}-\delta^K) u_1(a^*,b^*) > 0.
\end{equation}
Dividing the above expression by $1-\delta$, and then taking the limit where $\delta \rightarrow 1$, we obtain that inequality (\ref{4.24}) is true when $\delta$ is close to $1$ if
\begin{equation*}
    (K-m) \Big(u_1(a^*,b^*)-u_1(a^*,\beta)\Big)
   +mV_K -(m-1) u_1(a^*,b^*)-u_1(a',b^*)>0,
\end{equation*}
or equivalently,
\begin{equation}\label{4.25}
    \underbrace{(K-m)}_{ \geq 1} (1-\beta) \Big(u_1(a^*,b^*)-u_1(a^*,b')\Big) > (m-1) \underbrace{\Big(u_1(a^*,b^*)-V_K\Big)}_{<0} + \Big(
    u_1(a',b^*)-V_K
    \Big).
\end{equation}
When $\delta$ is close to $1$,
\begin{equation*}
    V_K \approx \frac{1}{K} u_1(a',b^*) + \frac{K-1}{K} u_1(a^*,b^*),
\end{equation*}
and therefore,
\begin{equation*}
    1-\beta = \frac{u_1(a',b^*)-V_K}{u_1(a',b^*)-u_1(a',b')} \approx \frac{K-1}{K} \cdot \frac{u_1(a',b^*)-u_1(a^*,b^*)}{u_1(a',b^*)-u_1(a',b')}
    > \frac{K-1}{K} \cdot \frac{u_1(a',b^*)-u_1(a^*,b^*)}{u_1(a^*,b^*)-u_1(a^*,b')},
\end{equation*}
where the last inequality follows from $u_1(a,b)$ having strictly increasing differences. This implies that
\begin{equation*}
    (1-\beta) \Big(u_1(a^*,b^*)-u_1(a^*,b')\Big)>u_1(a',b^*)-V_K \approx \frac{K-1}{K} \Big( u_1(a',b^*)-u_1(a^*,b^*) \Big).
\end{equation*}
Inequality (\ref{4.25}) is true when $\delta$ is close to $1$
since $K-m \geq 1$ and
$u_1(a',b^*)-V_K$ converges to $\frac{K-1}{K} (u_1(a',b^*)-u_1(a^*,b^*))$ as $\delta \rightarrow 1$.

Next, I show that player 1 has no incentive to play actions other than $a'$ and $a^*$ at every history that belongs to $\mathcal{H}_1^{**}$. It is straightforward to show that he has no incentive to play any action that does not belong to $\{a^*,a',\underline{a}\}$, since playing $\underline{a}$ leads to a strictly higher stage-game payoff for player 1 while not lowering his continuation value. Hence, I only need to show that when $a' \neq \underline{a}$, player 1 has no incentive to play $\underline{a}$ at any history that belongs to $\mathcal{H}_1^{**}$. This is because his payoff at any history that occurs with positive probability is bounded from below by $V_1 \approx \frac{1}{K} u_1(a',b^*) + \frac{K-1}{K} u_1(a^*,b^*)> u_1(a^*,b^*)$. Inequality (\ref{optimalcommitment}) implies that player 1's stage-game payoff is strictly less than $u_1(a^*,b^*)$ when player 2's action is strictly lower than $b^*$. Since player 2 plays $b''$ when actions other than $a^*$ and $a'$ occurred in the last $K$ periods, player 1's continuation value when he plays $\underline{a}$ is at most:
\begin{equation*}
    (1-\delta) u_1(\underline{a},b^*) + (\delta-\delta^2) u_1(a',b')
    +(\delta^2-\delta^{K+1}) u_1(a^*,b')
    +\delta^{K+1} V_K
\end{equation*}
Since $V_1<V_2<...<V_K$, it is sufficient to show that
\begin{equation*}
    V_1=(1-\delta^{K-1}) u_1(a^*,b^*)+\delta^{K-1} V_K  > (1-\delta) u_1(\underline{a},b^*) + (\delta-\delta^2) u_1(a',b')
    +(\delta^2-\delta^{K+1}) u_1(a^*,b')
    +\delta^{K+1} V_K.
\end{equation*}
or equivalently,
\begin{equation*}
    (1-\delta^{K-1}) u_1(a^*,b^*)+(\delta^{K-1}-\delta^{K+1}) V_K  - (1-\delta) u_1(\underline{a},b^*)- (\delta-\delta^2) u_1(a',b')
    -(\delta^2-\delta^{K+1}) u_1(a^*,b')>0.
\end{equation*}
Dividing the left-hand-side of the above inequality by $1-\delta$ and then taking the $\delta \rightarrow 1$ limit, we know that
the above inequality is true when $\delta$ is close to $1$ if
\begin{equation}\label{4.26}
    (K-1) u_1(a^*,b^*)+ 2V_K \geq u_1(\underline{a},b^*)+u_1(a',b')+(K-1) u_1(a^*,b').
\end{equation}
Since $u_1$ has strictly increasing differences, we have:
\begin{equation*}
    u_1(\underline{a},b^*) < u_1(a',b^*)-u_1(a',b')+u_1(\underline{a},b') \leq  u_1(a',b^*)-u_1(a',b')+ u_1(a^*,b^*)
\end{equation*}
\begin{equation*}
  <  u_1(a^*,b^*)+u_1(a',b')-u_1(a^*,b')-u_1(a',b')+ u_1(a^*,b^*)=2u_1(a^*,b^*)-u_1(a^*,b').
\end{equation*}
So the right-hand-side of (\ref{4.26}) is bounded from above by
    $2 u_1(a^*,b^*)+(K-2) u_1(a^*,b') +u_1(a',b')$, which is strictly less than $(K+1) u_1(a^*,b^*)$. Since $V_K > u_1(a^*,b^*)$, the left-hand-side of (\ref{4.26}) is strictly greater than $(K+1) u_1(a^*,b^*)$. This establishes inequality (\ref{4.26}).

In the last step, I show that if any action other than $a^*$ and $a'$ occurred in period $t-1$, player 1 has an incentive to play $a'$ in period $t$. Since $a_{t-1} \notin \{a^*,a'\}$, player 2's actions from period $t$ to period $t+K-1$ are $b''$ regardless of player 1's behavior in those periods, and moreover, player 2 has an incentive to play actions greater than $b''$ in period $s (\geq t+K)$ only if player 1 has played $a^*$ at least $K-1$ times and $a'$ at least once after the last time they played actions other than $a^*$ and $a'$. Since $V_K>V_{K-1}>...>V_1$, player 1's continuation value in period $t$ is bounded from above by:
\begin{equation*}
    (1-\delta) u_1(a',b'') +(\delta-\delta^K) u_1(a^*,b'') +\delta^K V_K.
\end{equation*}
This upper bound is attained when player 1 plays $a'$ in period $t$ and plays $a^*$ in the next $K-1$ periods, after which play reaches state $s_K$ and player 1's continuation value is $V_K$. This verifies his incentive to play $a'$ when his previous period action was neither $a^*$ nor $a'$.

\paragraph{Case 2:} I consider the case where $(u_1,u_2)$ violates (\ref{optimalcommitment}). Then there exist $a' \neq a^*$ and $b'$ (notice that $b'$ may equal $b^*$) such that $b'$ best replies to $a'$ and $u_1(a',b') \geq \max_{a \in A} \max_{b \in \textrm{BR}_2(a)} u_1(a,b)$. By definition, $u_1(a',b') \geq u_1(a^*,b^*)$. I construct equilibria where the rational-type of player $1$ plays $a'$ in every period and for every $t \geq 1$, player $2_t$ plays $b^*$ if $a^*$ was played in each of the last $K$ periods and
plays $b'$ if the former is not the case and no action except for $a^*$ and $a'$ appeared in the last $\min \{t,K\}$ periods.
The rest of the construction considers two subcases separately.

If $a'$ is player $1$'s lowest action, then at every history that occurs with positive probability, player $1$ plays $a'$ and player $2$ plays $b'$. Obviously, player $1$ has no incentive to play actions other than $a'$ at any history and player 2's strategy is also optimal given her belief, which verifies that this is an equilibrium.

If $a'$ is not player $1$'s lowest action, then let $a''$ be player $1$'s lowest action. By definition, there exists $\phi \in (0,1)$ such that $b'$ as well as an action strictly lower than $b'$, denoted by $b''$, are both best replies to $\alpha \equiv \phi a' + (1-\phi) a''$.
Upon observing any history that occurs with zero probability, if there exists any action that is neither $a'$ nor $a^*$, player $2$ believes that it occurred in the period before.
At every history that occurs with zero probability, player $1$ plays $a'$ if $a_{t-1} \in \{a^*,a'\}$ and plays $\alpha$ if $a_{t-1} \notin \{a^*,a'\}$.
Since player $2$ assigns probability $1$ to $a_{t-1} \notin \{a',a^*\}$ when she observes at least one action that is not $a'$ and $a^*$, she has an incentive to
mix between $b'$ and $b''$, where his probability of playing $b'$ is denoted by $\beta$ and is given by
\begin{equation}\label{14}
\beta u_1(a'',b') + (1-\beta) u_1(a'',b'')= (1-\delta^{K}) \Big(\beta u_1(a',b')+(1-\beta) u_1(a',b'')\Big) + \delta^K u_1(a',b').
\end{equation}
This implies that at a history where $a_{t-1} \notin \{a^*,a'\}$, player $1$ is indifferent between playing $a'$ and $a''$, and given that $u_1(a^*,b^*) \leq u_1(a',b')$, he strictly prefers $a''$ to any action that is not $a'$ or $a''$. What remains to be verified is that at a history where the last $K$ actions were $a'$, player $1$ has no incentive to play actions other than $a'$. First, playing $a^*$ is suboptimal given that $u_1(a^*,b^*) \leq u_1(a',b')$ and playing actions other than $a^*$ and $a'$ is strictly dominated by playing $a''$. Hence, I only need to verify that player $1$ has no incentive to play $a''$. Player $1$'s payoff when he plays $a''$ at history $(a_{t-K},...,a_{t-1})=(a',...,a')$ is
\begin{equation}\label{15}
  (1-\delta) u_1(a'',b') + \delta (1-\delta^{K}) \Big(\beta u_1(a',b')+(1-\beta) u_1(a',b'')\Big) + \delta^{K+1} u_1(a',b'),
\end{equation}
which by the definition of $\beta$ in (\ref{14}) as well as the assumption that $u_1$ has strictly increasing differences, implies that (\ref{15}) is strictly smaller than $u_1(a',b')$. This verifies that player 1 has no incentive to deviate.


\begin{thebibliography}{99}
\bibitem{pa} Acemoglu, Daron, Ali Makhdoum, Azarakhsh Malekian and Asu Ozdaglar (2022) ``Learning from Reviews: The Selection Effect and the Speed of Learning,'' \textit{Econometrica}, 90, 2857-2899.
\bibitem{pa} Acemoglu, Daron and Alexander Wolitzky (2014) ``Cycles of Conflict: An Economic Model,'' \textit{American Economic Review}, 104, 1350-1367.
\bibitem{pa} Batman, Daniel and Laura  Shaw (1991) ``Evidence for Altruism: Toward a Pluralism of Prosocial Motives,'' \textit{Psychological Inquiry}, 2, 107-122.
\bibitem{pa} Bhaskar, V. and Caroline Thomas (2019) ``Community Enforcement of Trust with Bounded Memory,'' \textit{Review of Economic Studies}, 86, 1010-1032.
\bibitem{pa} Clark, Daniel, Drew Fudenberg and Alexander Wolitzky (2021) ``Record-Keeping and Cooperation in Large Societies,'' \textit{Review of Economic Studies}, 88, 2179-2209.
\bibitem{pa} Cripps, Martin, George Mailath and Larry Samuelson (2004) ``Imperfect Monitoring and Impermanent Reputations,'' \textit{Econometrica}, 72, 407-432.
\bibitem{pa} Cripps, Martin and Caroline Thomas (2019) ``Strategic Experimentation in Queues,'' \textit{Theoretical Economics}, 14, 647-708.
\bibitem{pa} Deb, Joyee (2020) ``Cooperation and Community Responsibility,'' \textit{Journal of Political Economy}, 128, 1976-2009.
\bibitem{pa} Deb, Joyee, Takuo Sugaya and Alexander Wolitzky (2020) ``The Folk Theorem in Repeated Games With Anonymous Random Matching,'' \textit{Econometrica}, 88, 917-964.
\bibitem{pa} Dellarocas, Chrysanthos (2006) ``Reputation Mechanisms'' Handbook on Information Systems and Economics, T. Hendershott (ed.), Elsevier Publishing, 629-660.
\bibitem{pa} Ekmekci, Mehmet (2011) ``Sustainable Reputations with Rating Systems,'' \textit{Journal of Economic Theory}, 146, 479-503.
\bibitem{pa} Ekmekci, Mehmet, Olivier Gossner and Andrea Wilson (2012) ``Impermanent Types and Permanent Reputations,'' \textit{Journal of Economic Theory}, 147, 162-178.
\bibitem{pa} Ekmekci, Mehmet and Lucas Maestri (2022) ``Wait or Act Now? Learning Dynamics in Stopping Games,'' \textit{Journal of Economic Theory}, 205, 105541.
\bibitem{pa} Ellison, Glenn (1994) ``Cooperation in the Prisoner's Dilemma with Anonymous Random Matching,'' \textit{Review of Economic Studies}, 61, 567-588.
\bibitem{pa} Ely, Jeffrey and Juuso V\"{a}lim\"{a}ki (2003) ``Bad Reputation,'' \textit{Quarterly Journal of Economics}, 118, 785-814.
\bibitem{pa} Fudenberg, Drew and David Levine (1989) ``Reputation and Equilibrium Selection in Games with a Patient Player,'' \textit{Econometrica}, 57, 759-778.
\bibitem{pa} Fudenberg, Drew and David Levine (1992) ``Maintaining a Reputation when Strategies are Imperfectly Observed,'' \textit{Review of Economic Studies}, 59, 561-579.
\bibitem{pa} Fudenberg, Drew and Jean Tirole (1991) ``Perfect Bayesian Equilibrium and Sequential Equilibrium,'' \textit{Journal of Economic Theory}, 53, 236-260.
\bibitem{pa} Gossner, Olivier (2011) ``Simple Bounds on the Value of a Reputation,'' \textit{Econometrica}, 79, 1627-1641.
\bibitem{pa} Heller, Yuval and Erik Mohlin (2018) ``Observations on Cooperation,'' \textit{Review of Economic Studies}, 88, 1892-1935.
\bibitem{pa} Hu, Ju (2020) ``On the Existence of the Ex Post Symmetric Random Entry Model,'' \textit{Journal of Mathematical Economics}, 90, 42-47.
\bibitem{pa} Jehiel, Philippe and Larry Samuelson (2012) ``Reputation with Analogical Reasoning,'' \textit{Quarterly Journal of Economics}, 127(4), 1927-1970.
\bibitem{pa} Kandori, Michihiro (1992) ``Social Norms and Community Enforcement,'' \textit{Review of Economic Studies}, 59, 63-80.
\bibitem{pa} Kaya, Ay\c{c}a and Santanu Roy (2022) ``Market Screening with Limited Records,'' \textit{Games and Economic Behavior}, 132, 106-132.
\bibitem{pa} Levine, David (1998) ``Modeling Altruism and Spitefulness in Experiments,'' \textit{Review of Economic Dynamics}, 1, 593-622.
\bibitem{pa} Levine, David (2021) ``The Reputation Trap,'' \textit{Econometrica}, 89,  2659-2678.
\bibitem{pa} Li, Yingkai and Harry Pei (2021) ``Equilibrium Behaviors in Repeated Games,'' \textit{Journal of Economic Theory}, 193, 105222.
\bibitem{pa} Liu, Qingmin (2011) ``Information Acquisition and Reputation Dynamics,'' \textit{Review of Economic Studies}, 78, 1400-1425.
\bibitem{pa} Liu, Qingmin and Andrzej Skrzypacz (2014) ``Limited Records and Reputation Bubbles,'' \textit{Journal of Economic Theory} 151, 2-29.
\bibitem{pa} Mailath, George  and Larry Samuelson (2001) ``Who Wants a Good Reputation?'' \textit{Review of Economic Studies}, 68, 415-441.
\bibitem{pa} Mailath, George  and Larry Samuelson (2015) ``Reputations in Repeated Games,'' \textit{Handbook of Game Theory with Economic Applications}, 165-238.
\bibitem{pa} Pei, Harry (2020) ``Reputation Effects under Interdependent Values,'' \textit{Econometrica}, 88(5), 2175-2202.
\bibitem{pa} Pei, Harry (2022) ``Reputation Building under Observational Learning,'' \textit{Review of Economic Studies}, forthcoming.
\bibitem{pa} Quah, John and Bruno Strulovici (2012) ``Aggregating the Single Crossing Property,'' \textit{Econometrica}, 80, 2333-2348.
\bibitem{pa} Renault, J\'{e}r\^{o}me, Eilon Solan and Nicolas Vieille (2013) ``Dynamic Sender-Receiver Games,'' \textit{Journal of Economic Theory}, 148, 502-534.
\bibitem{pa} Sorin, Sylvain (1999) ``Merging, Reputation, and Repeated Games with Incomplete Information,'' \textit{Games and Economic Behavior}, 29, 274-308.
\bibitem{pa} Sugaya, Takuo and Alexander Wolitzky (2020) ``Do a Few Bad Apples Spoil the Barrel?: An Anti-Folk Theorem for Anonymous Repeated Games with Incomplete Information,'' \textit{American Economic Review}, 110, 3817-3835.
\bibitem{pa} Takahashi, Satoru (2010) ``Community Enforcement When Players Observe Partners' Past Play,'' \textit{Journal of Economic Theory}, 145, 42-62.
\bibitem{pa} Vong, Allen (2022) ``Certification for Consistent Quality Provision,'' Working Paper.
\end{thebibliography}
\end{document}